\documentclass[sigconf]{acmart}
\acmYear{2026}\copyrightyear{2026}
\setcopyright{cc}
\setcctype[4.0]{by}
\acmConference[SPAA '26]{38th ACM Symposium on Parallelism in Algorithms and Architectures}{July 6--10, 2026}{London, United Kingdom}
\acmBooktitle{38th ACM Symposium on Parallelism in Algorithms and Architectures (SPAA '26), July 6--10, 2026, London, United Kingdom}
\acmDOI{10.1145/3816782.3819182}
\acmISBN{979-8-4007-2761-0/26/07}

\makeatletter
\gdef\@copyrightpermission{
 \href{https://creativecommons.org/licenses/by/4.0/}{\includegraphics[width=0.15\columnwidth]{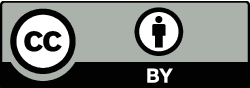}}\\[-1pt]
 \href{https://creativecommons.org/licenses/by/4.0/}{This work is licensed under a Creative Commons Attribution International 4.0 License.}
}
\makeatother

\usepackage{booktabs}   
\usepackage{subcaption} 
\usepackage{caption}
\usepackage{colortbl}
\usepackage{bigstrut}
\usepackage{microtype}
\usepackage{multirow}
\usepackage{multicol}
\usepackage{rotating}
\usepackage{lipsum}
\usepackage{balance}
\usepackage{mfirstuc}
\usepackage{titlecaps}
\usepackage{listings}
\usepackage{float}
\usepackage[rightcaption]{sidecap}
\usepackage{xcolor}
\usepackage[font=small,labelfont=bf,skip=2pt]{caption}
\usepackage{tabularx}

\def\fullversion{}

\usepackage{etoolbox}
\newcommand{\ifconference}[1]{{{\ifx\fullversion\undefined{#1}\fi}\xspace}}
\newcommand{\iffullversion}[1]{{{\ifx\conference\undefined{#1}\fi}\xspace}}

\usepackage{graphicx}  
\usepackage{lipsum}  
\newcommand{\hide}[1]{} 
\usepackage{xspace}
\usepackage{textcomp}
\usepackage{comment} 
\usepackage{verbatim}
\usepackage{url}

\newcommand{\defn}[1]{\emph{\boldmath\textbf{#1}\unboldmath}} 
\newcommand{\emp}[1]{{\boldmath\emph{\textbf{#1}}\unboldmath}} 
\newcommand{\fname}[1]{\textit{#1}} 

\newcommand{\fullbf}[1]{{\textbf{{\boldmath #1 \unboldmath}}}} 

\newtheorem{theorem}{Theorem}[section]
\newtheorem{lemma}{Lemma}[section]

\let \originalleft \left
\let\originalright\right
\renewcommand{\left}{\mathopen{}\mathclose\bgroup\originalleft}
\renewcommand{\right}{\aftergroup\egroup\originalright}

\usepackage{scalerel} 

\newtheoremstyle{exampstyle}
{.5em} 
{.5em} 
{\it} 
{.5em} 
{\sc} 
{.} 
{.5em} 
{} 
\theoremstyle{exampstyle} 
\theoremstyle{exampstyle} \newtheorem{compactlem}[lemma]{Lemma}
\theoremstyle{exampstyle} 
\theoremstyle{exampstyle} \newtheorem{compactthm}[theorem]{Theorem}
\theoremstyle{exampstyle} 

\makeatletter
\renewenvironment{proof}[1][\proofname]{\par
\vspace{-1\topsep}
\pushQED{\qed}%
\normalfont
\topsep0pt \partopsep0pt 
\trivlist
\item[\hskip\labelsep
      \itshape
  #1\@addpunct{.}]\ignorespaces
}{%
\popQED\endtrivlist\@endpefalse
}

\newenvironment{singleproof}[1][\proofname]{\par
\pushQED{\qed}%
\normalfont
\topsep0pt \partopsep0pt 
\trivlist
\item[\hskip\labelsep
      \itshape
  #1\@addpunct{.}]\ignorespaces
}{%
\popQED\endtrivlist\@endpefalse
}

\newcommand{\R}{\mathbb{R}}

\newcommand{\whp}[1]{\emph{whp}}

\DeclareMathOperator*{\polylog}{polylog}

\usepackage{pifont}

\newcommand{\modelop}[1]{\texttt{#1}}
\newcommand{\forkins}{\modelop{fork}}

\newcommand{\thread}{thread}

\usepackage[shortlabels]{enumitem}

\setlist{topsep=0.3em,itemsep=0.2em,parsep=0.1em,leftmargin=*}

\usepackage{float}
\usepackage[font={small},aboveskip=0.3em, belowskip=0.2em]{caption}

\usepackage[labelfont=bf,list=true,skip=0em]{subcaption}
\usepackage[rightcaption]{sidecap}

\usepackage{wrapfig}

\usepackage{array}
\newcolumntype{L}[1]{>{\raggedright\let\newline\\\arraybackslash\hspace{0pt}}m{#1}}
\newcolumntype{C}[1]{>{\centering\let\newline\\\arraybackslash\hspace{0pt}}m{#1}}
\newcolumntype{R}[1]{>{\raggedleft\let\newline\\\arraybackslash\hspace{0pt}}m{#1}}
\newcolumntype{B}{>{\bf}c}
\usepackage{rotating}

\usepackage{booktabs} 
\usepackage{multicol,multirow}
\usepackage{longtable} 
\usepackage{supertabular} 
\usepackage{colortbl}
\usepackage{bigstrut}

\usepackage{titlesec}

\titlespacing{\section}{0pt}{0.5em}{0.2em} 
\titlespacing{\subsection}{0pt}{0.5em}{0.2em} 

\newcommand{\myparagraph}[1]{\smallskip\noindent\emp{#1} \ }

\usepackage[ruled,lined,linesnumbered,noend]{algorithm2e}
\usepackage[noend]{algpseudocode}
\SetAlCapFnt{\small}
\SetAlCapNameFnt{\small}

\makeatletter
\patchcmd{\@algocf@start}
  {-1.5em}
  {0pt}
  {}{}
\newcommand{\nosemic}{\renewcommand{\@endalgocfline}{\relax}}
\newcommand{\dosemic}{\renewcommand{\@endalgocfline}{\algocf@endline}}

\SetSideCommentLeft
\SetKwInput{notations}{Notations}
\SetKwInput{notes}{Notes}
\SetKwInput{maintains}{Maintains}

\SetKwProg{myfunc}{Function}{}{}
\SetKwFor{parForEach}{ParallelForEach}{do}{endfor}
\SetKwFor{Justrepeat}{Repeat}{}{}
\SetKwFor{parDo}{In parallel:}{}{}

\SetKw{MIN}{min}
\SetKw{MAX}{max}
\SetKw{OR}{or}
\SetKw{AND}{and}

\SetCommentSty{mycommfont}

\usepackage{mdframed}
\mdfdefinestyle{mystyle}{innertopmargin=2pt,innerbottommargin=2pt,skipabove=2pt,skipbelow=2pt}
\mdfdefinestyle{densestyle}{linecolor=framelinecolor,innertopmargin=0,innerbottommargin=0,leftmargin=0,rightmargin=0,backgroundcolor=gray!20}
\mdfdefinestyle{compactcode}{linecolor=framelinecolor,innertopmargin=1pt,innerbottommargin=1pt,backgroundcolor=gray!20,skipabove=0pt,skipbelow=0pt,leftmargin=0,rightmargin=0}

\usepackage{framed}

\usepackage{listings}

\newdimen\zzsize
\newdimen\kwsize
\newcommand{\basicstyle}{\fontsize{\zzsize}{1\zzsize}\ttfamily}
\newcommand{\keywordstyle}{\fontsize{\kwsize}{1\kwsize}\ttfamily\bf}

\newdimen\zzlstwidth
\usepackage{cleveref}
\crefname{section}{Sec.}{Sec.}
\crefname{theorem}{Thm.}{Thm.}
\crefname{lemma}{Lemma}{Lemma}
\crefname{corollary}{Cor.}{Cor.}
\crefname{table}{Tab.}{Tab.}
\crefname{algorithm}{Alg.}{Alg.}
\crefname{figure}{Fig.}{Fig.}
\crefname{fact}{Fact}{Fact}
\Crefname{table}{Tab.}{Tab.}
\crefname{problem}{Problem}{Problem}
\crefname{line}{Line}{Line}

\usepackage{tikz} 

\setcopyright{cc}

\newcommand{\naive}{na\"ive}

\newcommand{\xiangyun}[1]{{\color{blue}{\bf Xiangyun:} #1}}

\newcommand{\msl}{metric skip list\xspace}
\newcommand{\Msl}{Metric Skip List\xspace}

\newcommand{\funcfont}[1]{{\textsf{#1}}}

\newcommand{\somef}{F}

\newcommand{\up}[1]{#1.\mathit{up}}
\newcommand{\down}[1]{#1.\mathit{down}}
\newcommand{\listcontent}[1]{#1.\mathit{list}}
\newcommand{\radius}[1]{#1.\mathit{radius}}
\newcommand{\rghttxt}{advance}
\newcommand{\rght}[1]{#1.\mathit{\rghttxt}}
\newcommand{\algntxt}{align}
\newcommand{\algn}[2]{\mathit{\algntxt}(#1,#2)}
\newcommand{\listoflist}[1]{\mathcal{F}_{#1}}
\newcommand{\findlistradius}[2]{\listoflist{#1}(#2)}
\newcommand{\findlistrank}[2]{\listoflist{#1}[#2]} 

\newcommand{\focuson}{focus on}

\newcommand{\focuspoint}{focus point}
\newcommand{\focuslist}{focus list}
\newcommand{\control}{control} 

\newcommand{\cur}{\mathit{cur}}
\newcommand{\nxt}{\mathit{nxt}} 
\newcommand{\last}{\mathit{last}}

\newcommand{\node}[1]{\mathit{node}(#1)}
\newcommand{\parent}[1]{\mathit{parent}(#1)}

\begin{document}

\title{
Parallel Metric Skip Lists and Nearest Neighbor Search
}

\author{Xiangyun Ding}
\affiliation{
  \institution{University of California, Riverside}
  \country{}
}
\email{xding047@ucr.edu}

\author{Rohin Garg}
\affiliation{
  \institution{Massachusetts Institute\\ of Technology}
  \country{}
}
\email{rohin@mit.edu}

\author{Yan Gu}
\affiliation{
  \institution{University of California, Riverside}
  \country{}
}
\email{ygu@cs.ucr.edu}

\author{Yihan Sun}
\affiliation{
  \institution{University of California, Riverside}
  \country{}
}
\email{yihans@cs.ucr.edu}

\renewcommand{\shortauthors}{Ding et al.}

\begin{abstract}

The \msl{} is a data structure designed for efficient nearest and $k$-nearest neighbor search in metric spaces. 
For many real-world datasets with reasonable distributions---specifically, those with a \emph{constant expansion rate}---it supports $\tilde{O}(n)$ construction time and $O(k\log n)$ query time, where $n$ is the input size and $k$ is the number of nearest neighbors in queries. 
Notably, unlike alternative approaches, it does not require a bounded aspect ratio, making it more flexible for input data distributions. 
However, the inherently sequential nature of its original construction has, to our knowledge, precluded any existing parallel algorithm. 

In this paper, we present highly parallel and work-efficient algorithms for constructing metric skip lists. Under the assumption of a constant expansion rate, our approach achieves an expected work of $O(n \log n)$ and a polylogarithmic span with high probability. Our design is based on novel algorithmic insights that improves the sequential procedure, enabling a divide-and-conquer strategy that facilitates parallelism while maintaining efficiency. 


With our algorithms, we can also support improved bounds for relevant applications using nearest neighbor as building blocks, including bichromatic closest pair (BCP), density-based clustering, and $k$-NN graph construction, among others. To our knowledge, many of these results represent the first solutions to achieve both work efficiency and polylogarithmic span, relying solely on the assumption of a constant expansion rate. 
\end{abstract}




\renewcommand\footnotetextcopyrightpermission[1]{} 
\fancyhead{} 


\maketitle



\begin{table*}[!t]
  \centering
  \small
  \setlength{\tabcolsep}{4pt}
  
  \newcolumntype{C}{>{\centering\arraybackslash}X}
  \newcolumntype{M}[1]{>{\centering\arraybackslash}m{#1}}

  \begin{tabularx}{\textwidth}{@{} M{0.3\textwidth} M{0.33\textwidth} C @{}}
    \toprule
     & \textbf{Original Cover Tree and Its Parallel Version} & \textbf{\Msl{s}} \\
    \midrule
    
    Proposed in & Beygelzimer et al.~\cite{beygelzimer2006cover}; analysis later corrected by Curtin~\cite{curtin2015improving} and Elkin and Kurlin~\cite{elkin2021new} & Karger and Ruhl~\cite{karger2002finding} \\
    \midrule
    
    Assumptions for the stated bounds & ~~~\quad Constant expansion rate, and\newline bounded aspect ratio & Only constant expansion rate \\
    \midrule
    
    \multirow{2}{*}{Sequential Construction Time} & \multirow{2}{*}{$O(n\log n)$} & $O(n\log n\log\log n)$ \whp{}~\cite{karger2002finding}. \\
     & & Improved to $O(n\log n)$ expected \textbf{[this paper]} \\
    \midrule
    
    \multirow{2}{*}{Nearest Neighbor Search Time} & \multirow{2}{*}{$O(\log n)$} & $O(\log n\log\log n)$ \whp{}~\cite{karger2002finding}. \\
     & & Improved to $O(\log n)$ expected \textbf{[this paper]} \\
    \midrule
    
    Parallel Construction Work & $O(n\log n)$ expected~\cite{gu2022parallel} & $O(n\log n)$ expected \textbf{[this paper]} \\
    \midrule
    
    Parallel Construction Span & $O(\log^3 n\log\log n)$ \whp{}~\cite{gu2022parallel} & $O(\log^3 n)$ \whp{} \textbf{[this paper]} \\
    \bottomrule
  \end{tabularx}
  \caption{Comparison between cover trees and \msl{s} for exact NNS in metric spaces. The stated bounds assume a constant expansion rate; the cover-tree bounds additionally use a bounded aspect ratio. For both data structures, the assumptions are needed for the displayed asymptotic bounds, not for correctness: without them, both still find exact nearest neighbors, but the expansion rate and the aspect ratio (for cover trees) appear in the cost bounds.}
  \label{tab:cover-tree-msl}
\end{table*}

\section{Introduction}

Nearest neighbor search (NNS) is a fundamental algorithmic primitive with extensive applications in computational geometry, computer graphics, machine learning, computer vision, and various other fields within artificial intelligence. 
In many practical scenarios, the distance between two points is defined using measures other than the standard Euclidean metric and $L_{p}$ norm.
Commonly employed alternatives include cosine similarity (related to angular distance), Jaccard distance~\cite{jaccard1901etude} (set-based similarity), tree/graph distances, Mahalanobis distance~\cite{mahalanobis1936generalized}, and Earth Mover's Distance~\cite{rubner1998metric}. 
All these distance measures satisfy specific properties (defined in \cref{sec:prelim}) and can be generalized to \emph{metric spaces}.

Many efficient data structures~\cite{karger2002finding,ruhl2003efficient,krauthgamer2004navigating,clarkson2006nearest,beygelzimer2006cover,izbicki2015faster,elkin2021new,gu2022parallel} have been developed specifically for NNS in metric spaces.
Unfortunately, finding (or even approximating) nearest neighbors in general metrics requires linear time~\cite{uniformmetric,beygelzimer2006cover}. 
Therefore, most of these data structures rely on assumptions regarding a ``good'' data distribution to achieve high efficiency. 
Two widely-studied properties include low expansion rate, which indicates that the density of points in the metric space changes smoothly, and bounded aspect ratio, 
which indicates that the scale of the dataset (maximum vs.\ minimum distance) is bounded.
We define these two concepts formally in \cref{sec:prelim}.

Under certain assumptions, the two canonical theoretical data structures for exact NNS in metric spaces are the \emph{metric skip list} introduced by Karger and Ruhl in 2002~\cite{karger2002finding} and the \emph{cover tree} proposed by Beygelzimer et al.~in 2006~\cite{beygelzimer2006cover}. 
Their best-known bounds are shown in \cref{tab:cover-tree-msl}.
The cover tree was long regarded as the standard theoretical solution for NNS in metric spaces since 1) its construction and query bounds are better by a factor of $O(\log\log n)$ where $n$ is the input size, and 2) it has a work-efficient parallel version with polylogarithmic span~\cite{gu2022parallel}.
However, in addition to a low expansion rate, the cover tree also assumes a bounded aspect ratio\footnote{A recent work~\cite{elkin2023new} showed a sequential cover tree variant without this assumption; to our knowledge, no parallel version of this variant or easy fix of~\cite{gu2022parallel} is known.} for both the original cover tree~\cite{beygelzimer2006cover} and its parallel version~\cite{gu2022parallel}. 
This assumption can be strong: in $L_p$ metrics, even a simple quad/octree can achieve the same construction and query bounds assuming a bounded aspect ratio~\cite{blelloch2022parallel}.
By contrast, \msl{s} only assume a constant expansion rate to achieve their theoretical guarantees. 

Therefore, this gap motivates a natural question for NNS: for \msl{s}, can we achieve the same construction time (work) as cover trees, as well as high parallelism (polylogarithmic span), while maintaining their flexibility of solely relying on the basic low expansion rate assumption? 
This would make \msl{} a theoretically powerful building block for NNS queries, enabling downstream applications to achieve strong bounds under minimal assumptions. 
However, designing such an efficient and parallel solution remained open. 
In this paper, we answer this question affirmatively. 
In particular, we present new algorithms for \msl{s} that achieve the same $O(n\log n)$ work as cover trees (modulo randomization) for construction with polylogarithmic span, as well as the same asymptotic query bounds, under weaker assumptions.

We note that achieving the strong cost bounds, both in work (overcoming the $O(\log\log n)$ overhead) and span (achieving parallelism), is highly non-trivial. 
Several factors contribute to this difficulty. 
The first reason is that the original sequential algorithm itself is substantially complicated.
The algorithm given in the original paper~\cite{karger2002finding,ruhl2003efficient} requires $O(n \log n \log \log n)$ work, which incurs an $O(\log\log n)$ factor overhead as compared to cover trees.
While that paper shared some high-level ideas regarding the possibility of achieving an efficient construction with $O(n \log n)$ work, no algorithmic details or formal analysis were provided. 
Hence, the first major contribution of this paper is to present an $O(n \log n)$-work sequential construction algorithm for \msl{s} and its analysis, which are presented in \cref{sec:seq}. 
As mentioned in~\cite{karger2002finding}, maintaining $O(n \log n)$ additional pointers can avoid the $O(\log \log n)$ overhead in queries. 
However, we note that achieving this requires a major redesign of the algorithm, as we show in detail in \cref{sec:seq}, since computing each of these pointers can require $O(\log n)$ work in the worst case. 
To bound the overall work in construction, we analyze and amortize these additional costs against other operations to achieve the improved bounds.

We note that achieving the strong cost bounds, both in work (overcoming the $O(\log\log n)$ overhead) and span (achieving parallelism), is highly non-trivial. 
Several factors contribute to this difficulty. 
The first reason is that the original sequential algorithm itself is substantially complicated.
The algorithm given in the original paper~\cite{karger2002finding,ruhl2003efficient} requires $O(n \log n \log \log n)$ work, which incurs an $O(\log\log n)$ factor overhead as compared to cover trees.
While that paper shared some high-level ideas regarding the possibility of achieving an efficient construction with $O(n \log n)$ work, no algorithmic details or formal analysis were provided. 
Hence, the first major contribution of this paper is to present an $O(n \log n)$-work sequential construction algorithm for \msl{s} and its analysis, which are presented in \cref{sec:seq}. 
As mentioned in~\cite{karger2002finding}, maintaining $O(n \log n)$ additional pointers can avoid the $O(\log \log n)$ overhead in queries. 
However, we note that achieving this requires a major redesign of the algorithm, as we show in detail in \cref{sec:seq}, since computing each of these pointers can require $O(\log n)$ work in the worst case. 
To bound the overall work in construction, we analyze and amortize these additional costs against other operations to achieve the improved bounds.

The second reason is the complicated computational structure of \msl{s} for both queries and construction. 
The high-level idea of \msl{} is to first permute all points randomly, obtaining the \emph{priority} for each point (lower rank means higher priority). 
We denote the permuted point with rank $i$ as $s_i$. 
For each point $s_i$, the goal is to maintain an index to retrieve a constant number $\alpha$ of random (specifically, highest-priority) points within any given distance $r$ from $s_i$. 
Such indices are referred to as \defn{finger lists}. 
Sequentially, the finger lists are built from the lowest-priority point ($s_n$) to the highest ($s_1$), one at a time. 
Each point $s_i$ maintains a sequence of finger lists: its first finger list simply includes the next $\alpha$ points after $s_i$; we denote this finger list as $F$. 
Then, we iterate through the rest of the points $s_j$ from $j=i+1$ to $n$. If $s_j$ is closer to $s_i$ than the current points in $F$, we replace the farthest point in $F$ with $s_j$, store the updated $F$ as the next finger list for $s_i$, and proceed to the next $j$. 
Due to the random permutation, Karger and Ruhl~\cite{karger2002finding} showed that each point has $O(\alpha\log n)$ finger lists \whp{}. 
These finger lists will navigate queries to find the nearest neighbor in a random-walk manner; a query for a point $q$ will start from $s_1$, using its finger list to identify a (roughly) closer point $s_{\nxt}$, and advance to the finger list of $s_{\nxt}$ to continue the process.

By definition, a trivial construction process takes $O(n)$ work per point ($O(n^2)$ total work), assuming $\alpha$ is a constant. 
The key insight for a more efficient construction algorithm is that, when constructing the finger list for $s_i$, we can consult the existing finger lists of $s_{i+1}\dots s_n$. 
At a high level, this process resembles a random walk that queries $s_i$ in the already constructed finger lists of $s_{i+1}\dots s_n$. 
Along the way, we evaluate the encountered points and build the finger lists for $s_i$ accordingly. 
This yields $\tilde{O}(n)$ construction work rather than the na\"{i}ve quadratic cost. 

Given this computational structure, work-efficient parallelization of this process is highly non-trivial. 
The iterative, element-by-element construction of the lists appears inherently sequential: 
computing the finger lists for $s_i$ may consult, and thus seemingly depend on, the finger lists of the subsequent $\alpha$ points. 
While the true dependencies are much sparser (computing the finger lists for $s_i$ may require one finger list of $s_j$ ($j>i$) but not all finger lists), 
analyzing the actual dependence graph remains very difficult due to the complications of both isolating these true dependencies and the randomization inherent in the algorithm. 

To overcome this obstacle, our solution is somewhat counterintuitive. 
We add artificial synchronization barriers to the dependence graph---in a divide-and-conquer manner.
For instance, at the top level, we compute the finger lists for points $s_1, \dots, s_{n/2}$ and the finger lists for points $s_{n/2+1}, \dots, s_n$ in parallel.
By definition, the second (right) subproblem is fully settled, but the first (left) one requires additional care to handle cross-boundary dependencies. 
Note that the random walk starts from the high-priority points and proceeds toward the lower-priority points.
Hence, for a point $s_i$ in the left subproblem, 
the random walk can be blocked without knowing the points on the right. 
This process must terminate when the random walk wishes to advance to the finger list of a point $s_j$ that has not been fully built. 
In this case, we say $s_j$ \defn{controls} $s_i$. 
By introducing the ``artificial barriers'' in the divide-and-conquer manner, we prove that
the ``control'' dependences among the input points form a tree, thereby facilitating more intuitive reasoning and analysis.
We show that the random-walk property guarantees the tree height as $O(\log n)$ \whp{}.
Combined with the divide-and-conquer paradigm, we show that the dependence depth for the entire computation is $O(\log^2 n)$ \whp{}, and the total span is $O(\log^3 n)$ \whp{}.
We provide more details of our algorithm and analysis in \cref{sec:parallel}.
Combining all these techniques, we present the main theorem of this paper below. 

\begin{theorem}
  Given a set of $n$ points with a constant expansion rate,
  the \msl{} can be constructed in $O(n \log n)$ expected work and $O(\log^3 n)$ span \whp{}.
  A nearest neighbor query can be answered in $O(\log n)$ expected time.
\end{theorem}

Nearest neighbor search is a fundamental primitive in computational geometry~\cite{abdelkader2021approximate,har2017proximity}, 
machine learning~\cite{chiu2019learning,chatzigeorgakidis2018fml}, 
and data mining~\cite{shakhnarovich2008nearest}.
Our results yield the best-known parallel theoretical bounds for many applications, 
including bichromatic closest pair, density-based clustering, and $k$-NN graph construction.
Due to space limitations, these applications are discussed in 
\ifconference{the full version of this paper \cite{metricskiplistfullversion}}\iffullversion{Appendix~\ref{sec:app:app}}.

\section{Preliminaries}\label{sec:prelim}

We use the term $O(f(n))$ \textbf{with high probability (\whp{})} in $n$ to indicate that the bound $O(g(k)\cdot f(n))$ holds with probability at least $1-1/n^k$ for \emph{any} $k\ge 1$ for some function $g(k)$. 
When the context is clear, we drop ``in $n$''.
Our notation is stronger than the high probability in some other work that only assumes a given $k$ that exists.
This is necessary in analyzing the expected work bound for our proposed algorithms.

\myparagraph{Computational Model.}
When analyzing the algorithms, we use the standard work-span model for fork-join parallelism with binary forking~\cite{CLRS,blelloch2020optimal}. 
We assume a set of \thread{}s that share a common memory.  Each \thread{} supports standard RAM instructions, a \forkins{}
instruction that forks two new child \thread{}s, and a constant-time test-and-set atomic operation.
When both children complete, the parent \thread{} continues.
A computation starts with a single {root} \thread{} and finishes when that root \thread{} finishes.
The \defn{work} of an algorithm is the total number of instructions. 
The \defn{span} (depth) is the length of the longest sequence of dependent instructions in the computation.
The randomized work-stealing scheduler can execute such a computation in $W/P+O(S)$ time \whp{} in $W$~\cite{BL98,arora2001thread,gu2022analysis}.
We say a parallel algorithm is \defn{work-efficient} if its work is $O(W)$,
where $W$ is the work of the best-known sequential algorithm.

\myparagraph{Low-Expansion Metric Space.}
The efficiency of many nearest-neighbor search algorithms and data structures relies on the assumption of low expansion rate of a metric space.
A metric $(X,d_X)$ is defined on a set $X$ with a distance function $d: X\times X \to \R^*$ that satisfies properties for any $x,y,z\in X$: 
\begin{enumerate}
  \item $d_X(x,y)=0\Leftrightarrow x=y$,
  \item $d_X(x,y) = d_X(y,x)$, and
  \item $d_X(x,y)\leq d_X(x,z)+d_X(z,y)$ (\emph{the triangle inequality}).
\end{enumerate}
When the context is clear, we drop the subscript $X$. We use \emph{point} or \emph{record} interchangeably to refer to the objects in the metric space. 

Define $B_X(p,r) = \{x \in X~|~d(p,x) \leq r\}$ as the closed ball centered at point $p$ and containing all points in $X$ at a distance of at most $r$ from $p$.
With clear context, we drop $X$.
We say a metric has \defn{$(\rho,c)$-expansion}~\cite{karger2002finding} \emph{iff} for all $p\in X$ and $r>0$,
\[
|B(p,r)| \geq \rho \implies |B(p,2r)| \leq c\cdot |B(p,r)|.
\]
The parameter $c$ is referred to as the \defn{expansion rate} of the metric space, and we say a metric has a low or constant expansion rate if $c=O(1)$.
Usually we assume $\rho$ is $O(\log|X|)$, which guarantees constant expansion for most real-world datasets.
Otherwise, the bounds for the find operation for the \msl{} will include an additive term of $O(\rho)$~\cite{ruhl2003efficient}.
Intuitively, low expansion means a smooth distribution of the points,
and rules out the case where, as the ball grows, we encounter a few points, then a long distance with no points, then a sudden and massive increase in the number of points.
This case is also less likely in most real-world datasets.

\myparagraph{Bounded Aspect Ratio.}
Although not used in \msl{s}, for completeness, we define the aspect ratio, which is assumed to be bounded for cover trees to achieve the same asymptotic bounds as \msl{s}.
Aspect ratio is defined as $$\Delta=\frac{\max\{d(x,y)~|~\,x,y\in X\}}{\min\{d(x,y)~|~\,x,y\in X\}}.$$
To achieve $O(\log n)$ work for construction, cover trees in~\cite{beygelzimer2006cover} and the parallel version in~\cite{gu2023parallel}, require the aspect ratio $\Delta\le n^\kappa$ for some constant $\kappa>0$ to achieve the theoretical bounds, while \msl{s} do not have such a constraint. 

\myparagraph{Notation.} In this paper, we study a set of elements in the metric space. 
A common scenario is to consider multi-dimensional points, so we use \emph{points} to refer to the elements for simplicity, 
but our algorithm works on a general metric space. 
Based on the nature of the \msl{} algorithm, we assume the input point set $S$ has been randomly shuffled, and use $s_i$ to denote
the $i$-th point. 
When the context is clear, we may also use the index (also the priority) $i$ to refer to the point $s_i$. We use $n=|S|$
as the input size. 

The notation used in this paper are summarized in \cref{tab:notation}. More details about the notations for \emph{finger lists} will be introduced in \cref{sec:msl}
along with the background of the \msl{}.

\begin{table}[t!]
\small
\centering
\renewcommand{\arraystretch}{1.25}
\begin{tabular}{@{ }l@{ }l@{}}
\hline
\multicolumn{2}{@{}l}{\textbf{General Terms}} \\ \hline
\fullbf{$S$}       & $S=\{s_{1..n}\}$: A randomly shuffled sequence of input points.                 \\ 
\fullbf{$c$}              & The constant expansion rate.                            \\ 
\fullbf{$d(p, q)$}    & The distance between points $p$ and $q$.            \\ 
\hline
\multicolumn{2}{@{}l}{\textbf{The Finger Lists}} \\ \hline
\fullbf{$\alpha$}           & The sample size, or the size of each finger list. \\ 
\boldmath $\listoflist{i}$ \unboldmath  & An array of finger lists for point $i$, reversely sorted by radii. \\ 
\boldmath $\findlistrank{i}{k}$ \unboldmath  & The $k$-th finger list in $\listoflist{i}$. \\ 
\boldmath $\findlistradius{i}{r}$ \unboldmath & Returns the first finger list in $\listoflist{i}$ with radius \emph{no more than} $r$. \\ 
&{This can be performed  in $O(\log |\listoflist{i}|)=O(\log\log n)$  time  }\\
&{\whp{} by binary search.}\\
\multicolumn{2}{@{}l}{\boldmath $\algn{\findlistrank{i}{k}}{r}$ \unboldmath \quad Returns $F^*=\findlistradius{i}{r}$, by traversing the array $\findlistrank{i}{\cdot}$ up}\\ 
&  or down from $\findlistrank{i}{k}$.\\
\multicolumn{2}{@{}l}{Let $\boldsymbol{F}=\findlistrank{i}{k}$ be the $k$-th of the finger lists of point $i$. It maintains:} \\ 
\boldmath ~~$\listcontent{\somef}$ \unboldmath & All points in this finger list. \\
&With clear context, we directly use $\somef$ to represent $\listcontent{F}$.\\
\boldmath ~~$\radius{\somef}$ \boldmath & The radius of the finger list, i.e., $\max_{j\in F}d(s_i,s_j)$.\\
\boldmath ~~$\rght{\somef}$ \boldmath & A mapping from each element in $F$ to another finger list. \\
&For each $j\in F$, $\rght{\somef}[j]$ is a pointer to the finger list\\
&  $\findlistradius{j}{r}$, where $r=\radius{\somef}$.\\
\hline

 \multicolumn{2}{@{}l}{\textbf{General Algorithm Descriptions}} \\ \hline
 \fullbf{$\cur$}        & (or $s_{\cur}$) The \emph{\focuspoint}. The current point processed by the  \\
 &construction/query algorithms, or the random walk scheme.\\
 \fullbf{$\nxt$ }        & (or $s_{\nxt}$) The next \focuspoint.  \\
 \fullbf{$F^*$}         & The \emph{\focuslist}. One finger list in $\listoflist{\cur}$, where we will find the  \\
 &next \focuspoint{} $\nxt$ from.\\
 \hline

\end{tabular}
\caption{Notations used in algorithms and analysis.}
\label{tab:notation}
\end{table}

\hide{
\begin{theorem}[Sampling Theorem~\cite{karger2002finding}]\label{thm:sampling} 
  Given a metric with $(\rho,c)$-expansion, a uniformly random subset $X'$ with $|X'|=m$ will have $(\max(c\rho,O(\log m)),2c)$-expansion with high probability to $m$.
\end{theorem}
We note that if we want to improve the probability to the size of~$n$, then the analysis in~\cite{karger2002finding} implies that we just change $O(\log m)$ in the parameter to $O(\log n)$.

This theorem shows that for low-expansion metric space, a sample of this space is also low-expansion.
This is crucial for our randomized batch-insertion algorithm in \cref{sec:insert}.

The query cost for cover trees relies on the expansion rate including or excluding the query point(s).
For simplicity, we still use $c$ to denote the expansion rate after such modifications, if it is increased.
In particular, for single-point query and bichromatic closest pairs, the query expansion rate includes the query point; for Euclidean MST and single-linkage clustering, the query expansion rate excludes the query points.

This theorem is important to improve the probabilistic bound on $\rho$; in the analysis in~\cite{??} it implies that we just change $O(\log m)$ in the parameter to $O(\log n)$.

This theorem also holds for low-expansion metric space, a sample of this space is also low-expansion. This is crucial for the complexity of construction and queries to the randomly sampled finger lists.
}

\section{The Metric Skip Lists}\label{sec:seq}

The \msl{} is a randomized data structure for nearest neighbor search, originally proposed by Karger and Ruhl~\cite{karger2002finding}.
In the same work, they introduced the concept of the expansion rate of a metric space (as detailed in \cref{sec:prelim}), which is often referred to as the KR-dimension. 
This metric has since been extensively used as a standard assumption for analyzing algorithms in metric spaces. 
Under the assumption of a constant expansion rate, Karger and Ruhl showed that an \msl{} of $n$ points can be constructed with $O(n\log n\log \log n)$ work, 
while enabling nearest neighbor queries in $O(\log n\log \log n)$ cost.

The main contribution of this paper is to improve the cost bounds of \msl{s}. 
Sequentially, we improve the work bounds for construction and query to $O(n\log n)$ and $O(\log n)$, respectively, which are presented in \cref{sec:work-efficient}. 
More importantly, we provide the first parallel version of the \msl{}, where the construction is work-efficient with polylogarithmic span, given in \cref{sec:parallel}.

Before presenting our contributions, we first overview the original construction of \msl{s} in this section. 
We start with the underlying intuition, followed by the formal definition and algorithms for \msl{s}. 
This background facilitates a better understanding of our improved algorithms and their analysis.

\subsection{The Random-Walk Procedure}\label{sec:intuition}

The \msl{} is designed for nearest neighbor search, motivated by the intuition of a random-walk procedure to reach the nearest neighbor of a query point $q$ among an input set $S$ using \emph{sampling}. 
At a high level, one starts from an arbitrary point $p \in S$ and iteratively attempts to walk to a closer point. 
By the triangle inequality, if a point closer to $q$ than $p$ exists, it must reside within the ball $B_S(p, 2d(p, q))$. 
Therefore, the \emph{sampling scheme} inspects a random sample of $\alpha'$ points in $B_S(p, 2d(p, q))$, 
moves to the one closest to $q$ as the next candidate, and repeats this process. 
This procedure is outlined in \cref{alg:intuition}. 
With an appropriately chosen sample rate, each step has a sufficient probability of ``making progress'' toward a significantly closer point, despite a small probability of moving to a farther one. 
Existing work showed that the following lemma holds by setting $\alpha' = 2c^3$ for an expansion rate $c$.

\begin{compactlem}[\cite{karger2002finding,ruhl2003efficient}]\label{lem:intuition}
  Given a metric space $(S,d_S)$ with expansion rate $c$, the random walk in \cref{alg:intuition} reaches the nearest neighbor of $q$ in $O(\log |S|)$ iterations \whp{}.
\end{compactlem}

\begin{algorithm}[t]
\small
\caption{The random-walk process to $q$'s nearest neighbor based on sampling~\cite{karger2002finding,ruhl2003efficient}\label{alg:intuition}}
\KwIn{A set of points \( S \) and a query point \( q \)}
\DontPrintSemicolon
Let \( p \) be an arbitrary point in \( S \)\\
\Repeat{\( p \) is the nearest neighbor of \( q \) in \( S \)}{
    Let \( S' \) be a random sample of \( 2c^3 \) points in \( B_S(p, 2d(p, q)) \)\label{line:sample}\\
    Update $p$ as the closest point to \( q \) in \( S'\)\\
}

\end{algorithm}

While a formal proof is provided in~\cite{karger2002finding,ruhl2003efficient}, we include a proof sketch here to facilitate the understanding of our improved algorithm and its analysis. 
To analyze this random-walk procedure, we define the \defn{rank} of a point $p \in S$ with respect to the query point $q$ as its position in the sorted order of distances to $q$. 
Specifically, $p$ has rank $k$ if it is the $k$-th closest point to $q$ in $S$. 
We will show that throughout the random walk, the rank of $p$ decreases rapidly, dropping to $1$ in $O(\log n)$ iterations \whp{}.

Let the current point being processed be $p$ with rank $k$. 
We denote $\delta = d(p, q)$ and let $B$ represent the ball $B_S(p, 2\delta)$.
We first prove $|B| \le c^2 k$. 
Since $p$ has rank $k$, there are at most $k$ points in $B_S(q, \delta)$. 
For any point $p' \in B$, the triangle inequality implies $d(q, p') \le d(q, p) + d(p, p') \le \delta + 2\delta < 4\delta$. 
Thus, all points in $B$ are contained in $B_S(q, 4\delta)$. 
By the definition of the expansion rate and the fact that $|B_S(q, \delta)| \le k$, we obtain $|B| \le |B_S(q, 4\delta)| \le c^2 k$. 

Therefore, the next point is selected by sampling $\alpha' = 2c^3$ points from at most $c^2 k$ candidates.
We show that, in this case, the closest point is likely to have a rank smaller than $k/2$. 
The probability that all $2c^3$ samples have rank at least $k/2$ is at most 
$\left(1 - \frac{k/2}{c^2 k}\right)^{2c^3} = \left(1 - \frac{1}{2c^2}\right)^{2c^3} \le e^{-c}$. 
Thus, the rank of the next candidate point is halved with probability at least $1 - e^{-c}$. 
In the complementary case, the candidate point remains in $B$, which has rank at most $c^2 k$. 
This implies that even in the worst case, the rank increases by a factor of at most $c^2$ with low probability $e^{-c}$. 
Combining the two cases, one can verify that the rank decreases by a constant factor in expectation for any $c\ge 1$, ensuring that the nearest neighbor is reached in $O(\log n)$ iterations \whp{}.

\hide{
Note that all points in $B_S(p, 2d(p, q))$ will have distance at most twice of $d(p, q)$, so their $q$-ranks are no more than $c^2 q_p$ since $|B| \le c^2 q_p$.
Meanwhile, with very large chance, the new point gets much closer to $q$, and here we mean by the new point has a $q$-rank smaller than $q_p/2$.
Given the sample size, this happens with probability at least $1-e^{-c}$.
Combining the two cases shows that with constant probability, the overall $q$-rank shrinks with a constant fraction in expectation in every iteration.
This guarantees that the random walk will reach the nearest neighbor of \( q \) in \( S \) in $O(\log|S|)$ iterations \whp{}.
}

\subsection{The Structure of \Msl{s}}\label{sec:msl}

With the background, we now introduce the \msl{s}.
The \msl{s} are built upon the sampling idea introduced in \cref{alg:intuition}. 
At a high level, we aim to build a structure such that for each point $s_i\in S$ and a given radius $r$, 
the data structure can easily look up $\findlistradius{i}{r}$ as a random sample of $\alpha$ points in ball $B_S(s_i,r)$. 
While $r$ can be an arbitrary real number, for each point, there are at most $n$ discrete values of $r$ that are meaningful.
Therefore, the idea is to randomly shuffle all the points, and for any of the meaningful values of $r$, we store the $\alpha$ points with the highest priorities. 

\myparagraph{Definition of \Msl{s} and Finger Lists.} More specifically, a \msl{} $\mathcal{F}$ maintains a structure $\listoflist{i}$ for the $i$-th point $s_i$, which is a series of \defn{finger lists}.
Each \emph{finger list} contains $\alpha$ points with its corresponding \emph{radius}, which is the largest distance among all points in this finger list to $s_i$. 
The first finger list $\findlistrank{i}{1}$ simply consists of the next $\alpha$ points, i.e., $\{s_{i+1},...,s_{i+\alpha}\}$. 
We then iterate through the rest of the points in order, and if a point is closer than the radius of the current list, we will replace the farthest point with the new point, 
create a new finger list, and make it the next finger list in $\listoflist{i}$.
In this paper, we call a point $j$ an \defn{evictor} of a finger list $\listoflist{i}[k]$, if $j$ is the point that successfully evicts the farthest point in $\listoflist{i}[k]$, thereby forming $\listoflist{i}[k+1]$. An example is shown in \cref{fig:fingle-list-example}.

After we iterate through all points, we continue the process of removing the farthest point in the finger list, and construct $\alpha$ more finger lists at the end. We call these finger lists the \defn{tail} in $\listoflist{i}$. 
This concept is useful for the correctness of the algorithms. 

\begin{figure}[t]
  \centering
  \includegraphics[width=0.9\columnwidth]{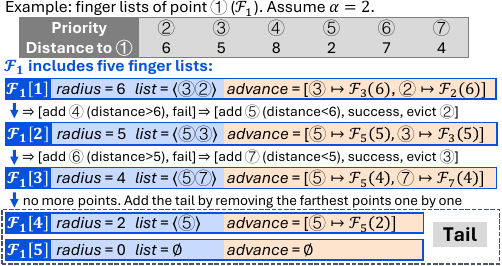}
  \caption{An example of the finger lists of Point 1. The blue fields are used in the original paper~\cite{karger2002finding,ruhl2003efficient}. 
  The orange fields are auxiliary metadata maintained in this paper. The definition of the fields are introduced in \cref{tab:notation}.}\label{fig:fingle-list-example}
\end{figure}

\myparagraph{Properties of a Finger List.} Although the definition may suggest $O(n)$ worst-case size for each finger list, 
due to the random permutation, the actual size is much smaller \whp{}. 
Existing work has shown the following lemma about the size of finger lists. 

\begin{compactlem}\label{lem:size}
  Given a randomly permuted point set $S=\{s_{1..n}\}$ and a constant $\alpha$, $|\listoflist{i}| = O(\alpha\log{n})$ \whp{}.
\end{compactlem}

It is easy to see that this lemma is true based on the random permutation.
In fact, $|\listoflist{i}|\le \alpha\cdot H_n$ where $H_n$ is the harmonic series. 

Note that for any point $i$, all finger lists in $\listoflist{i}$ are naturally sorted by distance (high to low). 
Therefore, looking for $\findlistradius{i}{r}$ can be easily performed by a binary search in $O(\log \log n)$ time. 
For each $i$ and $r$, $\findlistradius{i}{r}$ contains the $\alpha$ points with highest priority in radius $r$, forming a random sample set. 
In this paper, we set $\alpha=16c^3$. 
This is different from the parameter in \cref{alg:intuition}, and we will explain the reason later when we present the construction and query algorithms in \cref{sec:kr}.
We also note that we use a different $\alpha$ from the original paper~\cite{karger2002finding,ruhl2003efficient} to simplify the analysis, 
but in both cases $\alpha$ is a function of $c$ and is a constant with constant expansion rate. 


\myparagraph{The Organization and Notation of Finger Lists.} In the original paper, each finger list only maintains the corresponding list of $\alpha$ points and its radius. 
Therefore, we represent each finger list $\somef$ as a structure described in \cref{tab:notation}.
In particular, we use $\listcontent{\somef}$ to store the list of points. With clear context, we directly use $\somef$ as $\listcontent{\somef}$. 
We use $\radius{\somef}$ to represent the radius of $\somef$. 

In our paper, we will further augment each list with 
useful pointers (the \emph{.\rghttxt} pointers) to accelerate the algorithms. 
We introduce them in detail in \cref{sec:work-efficient,sec:parallel}. 
The major overhead of the original algorithm stems from using binary search to look up $\findlistradius{i}{r}$, which costs $O(\log\log n)$ \whp{} each time. 
Our augmented pointers are mainly used to avoid this cost, which improves the bound even for the existing sequential algorithm. 

Although the definition of \msl{} may suggest $O(n^2)$ work in the worst case for construction, 
since the total size of all finger lists is $O(\alpha n\log n)$ \whp{}, one can actually use a more efficient algorithm to build them. 
Next, we present the original construction and query algorithms proposed in Karger and Ruhl~\cite{karger2002finding} (referred to as the {KR algorithm}),
where construction takes $O(n\log n\log\log n)$ work and a nearest neighbor query takes $O(\log n\log\log n)$ work.

\subsection{The KR Algorithms}\label{sec:kr}

In this section, we review the original algorithm from~\cite{karger2002finding} for the \msl{s} introduced above.
With the \msl{s}, we can efficiently query the nearest neighbor, $k$-nearest neighbor, or range query (count or list the number of points within a given radius).
All query algorithms are similar, and thus we use nearest neighbor query as an example.

Note that the KR algorithms rely on the function $\findlistradius{i}{r}$ that finds the last finger list in $\listoflist{i}$ with radius no more than $r$, as introduced above.
This is achieved by binary searching in $\listoflist{i}$ and has $O(\log \log n)$ cost \whp{} due to \cref{lem:size}.
We will later show how to avoid using $\findlistradius{i}{r}$ in our algorithm in \cref{sec:work-efficient}.

\begin{algorithm}[t]
\caption{\fname{Construct}($s_{1 \ldots n}$) in~\cite{karger2002finding}\label{alg:orig-build}}
\label{alg:cap}
\small
\SetAlgoLined
\DontPrintSemicolon
\KwIn{A list of randomly permuted points $s_{1 .. n}$, and $\alpha$.}
\KwOut{The \msl{s} $\listoflist{1 .. n}$.}
\medskip

\SetKw{Break}{break}
\SetKw{downto}{downto}
\For{$i \gets n$ \downto $1$}{
    \fname{Build-Finger}($i$)\tcp*[r]{building the finger lists for the $i$-th point}
}

\smallskip

\myfunc{\upshape{\fname{Build-Finger}($i$)}}{

$\listcontent{\listoflist{i}[1]}\gets \{s_{i+1}, s_{i+2}, \dots, s_{\min\{i+\alpha,n\}}\}$\label{line:alg:krconstruct:first}\\
$\cur \gets i + \alpha$ \label{line:alg:krconstruct:cur}\\

\While{true}{
    $\somef{} \gets$ the last finger list of $s_i$ built so far\\
    $r \gets \max\{d(s_{\cur}, s_{i}) +\radius{\somef{}}, 2d(s_{\cur}, s_{i})\}$\label{line:alg:krconstruct:r}\\
    $F^*\gets\findlistradius{\cur}{r}$\\
    \If{\upshape $\exists s_j \in F^*$ s.t. $d(s_j, s_{i}) < \max\{\radius{\somef{}}, d(s_{\cur}, s_{i})\}$\label{line:fir}}{
        $\nxt \gets$ the first $j$ with this property in $F^*$\label{line:first}\\
        \If{$d(s_{\nxt}, s_{i}) < \radius{\somef{}}$}{
            $\somef{}_{\mathit{new}}\gets $ replacing the farthest point in $\somef{}$ with $s_{\nxt}$\label{line:append-start}\\
            Append $\somef{}_{\mathit{new}}$ to $\listoflist{i}$\label{line:append-end}\\
        }
    }
    \Else{
    \lIf{$|F^*|<\alpha$\label{line:alg:krconstruct:else}}{
        Build the tail and  \Break \label{line:alg:krconstruct:notfull}}
    \lElse{
        $\nxt \gets$ the last index in $F^*$\label{line:farther}
    }    
    }
    $\cur \gets \nxt$\\
}
}
\end{algorithm}

\myparagraph{The Construction Algorithm.}
The KR construction algorithm builds the \msl{} from $\listoflist{n}$ down to $\listoflist{1}$, in descending order of indexes. 
We present the algorithm in \cref{alg:orig-build}. 
Each $\listoflist{i}$ will leverage finger lists for those with indexes greater than $i$, which are already available. 
By doing so, the building $\listoflist{i}$ avoids an exhaustive search over all subsequent points. 
This will be achieved by searching $s_i$ in existing finger lists---following the random walk scheme in \cref{sec:intuition}---
and on the way, building all the finger lists in $\listoflist{i}[\cdot]$.

We now explain \cref{alg:orig-build} in detail. 
We start with the first finger list, which by definition contains the $\alpha$ points right after $s_{i}$ (line~\ref{line:alg:krconstruct:first}). 
Let this finger list be $F$. 
Throughout the algorithm, $F$ represents the last built finger list, and the goal is to build the next $F$.
The goal is then to find the \emph{evictor} of $F$, i.e., the first point after $s_{i+\alpha}$ with a smaller distance, within the random walk. 
As mentioned, the algorithm simulates a random walk to query $s_i$. 
We maintain a variable $\cur$ in \cref{alg:orig-build} to represent the current point in the random walk. 
Note that $\cur>i$ must be later than $i$ and have its finger lists ready. 
We call $\cur$ the \defn{\focuspoint{}} at this stage and say we \defn{\focuson{}} $\cur$. 
During the process, we will attempt to look for the evictor in $\listoflist{\cur}$, and find the next \focuspoint{} (the $\nxt$ variable in \cref{alg:orig-build}) for the random walk.
We start with $\cur=i+\alpha$ (line~\ref{line:alg:krconstruct:cur}). 

Based on the random walk scheme, we need to focus on a ball $B^*$ around $s_{\cur}$ with radius $r \ge 2d(s_{\cur},s_{i})$ for quick convergence. 
Moreover, if an evictor $j$ exists, it must satisfy $d(s_j,s_i)<\radius{F}$, and thus
we want the radius of $B^*$ to also cover $r\ge\radius{F}+d(s_i,s_{\cur})$ to avoid missing any evictor. 
Therefore, the radius $r$ is calculated in line~\ref{line:alg:krconstruct:r}.
By definition of finger lists, this ball $B^*$ corresponds to $\findlistradius{\cur}{r}$, which contains $\alpha$ random points within $B^*$. 
It can be easily located by a binary search on $\findlistrank{\cur}{\cdot}$. 
Let $F^*=\findlistradius{\cur}{r}$. 
Similarly, we call $F^*$ the \defn{\focuslist{}} at this stage of the algorithm, and say that we \defn{\focuson} $\findlistradius{\cur}{r}$. 

Among all points in the \focuslist{}, 
if any point is closer to $s_i$ than $s_{\cur}$ or than the radius of the last built finger list, 
we will set the one with the smallest index (highest priority) as the next \focuspoint{} $\nxt$. 
If $\nxt$ happens to be also within $\radius{F}$, it is exactly the evictor. 
Hence, we can construct the next finger list accordingly and append it to $\listoflist{i}$ (lines~\ref{line:append-start}--\ref{line:append-end}). 

If all points turn out to be farther (line~\ref{line:alg:krconstruct:else}), we first check whether $F^*$ is a tail list (i.e., contains fewer than $\alpha$ points). If so, we know we have exhausted the points in $B(s_{\cur},\radius{F})$, and the random walk terminates. 
Therefore, we build the tail and exit (line~\ref{line:alg:krconstruct:notfull}). 
If $F^*$ is a regular finger but no closer point is found, we set $\nxt$ to the last index in $F^*$ to continue this process (line~\ref{line:farther}). 
This indicates that the random walk temporarily goes to a farther point than $s_{\cur}$. 

The construction algorithm is illustrated in~\cref{fig:control-example}, which is presented with the parallel algorithm.
The first four subfigures (a)--(d) are equivalent to the construction process of point $i=1$ based on the already-constructed finger lists for points 2 to 7.
At a high level, the construction of the finger lists $\listoflist{i}$ for a point $s_{i}$ resembles the random-walk process in \cref{alg:intuition}. 
The series of \focuspoint{s} are the points visited in the random walk, and will ``advance'' gradually to the right (larger indexes/lower priorities) until we reach the end of the \msl{}. 
This is performed by maintaining the variables $\cur,F^*,\nxt$ in the algorithm. $\cur$ denotes the current focus point. 
$F^*$ is the focus list, which is a finger list of $\cur$. Finally, $\nxt$ is selected from $F^*$ as the next focus point.
The \focuspoint{} $s_{\cur}$ in each iteration can get closer to $s_{i}$ (line~\ref{line:first}) or farther (line~\ref{line:farther}), 
but will finally visit all ``useful'' points that should appear in $\listoflist{i}$. 

However, differently from \cref{alg:intuition}, the samples in \cref{alg:orig-build} are now decided by the priorities of a predefined random permutation.
Hence, the algorithm is also changed to jump to the first closer point (line~\ref{line:first}) rather than the closest point.
If no such points are found, we take the last point in $F^*$ (line~\ref{line:farther}) to continue advancing the \focuspoint{s}. 
As such, we do not rely on any randomness of the points in the future iterations, so the probability distributions of the random walk in different iterations remain independent.
Due to this change, we need a larger probability to find a closer point, so $\alpha$ is increased to $16c^3$ (or larger) to offset this difference~\cite{karger2002finding}.
This leads to the following result.

\begin{compactthm}[\cite{karger2002finding}]
    Given a randomly permuted point set $S=\{s_{1..n}\}$ with a constant expansion rate and a constant $\alpha$, the \msl{} can be constructed in $O(n\log{n}\log \log n)$ work \whp{}.
\end{compactthm}

\myparagraph{The Query Algorithms.}
The nearest neighbor query algorithm in \msl{s} is very similar to the construction algorithm (see \cref{alg:orig-find}).
The process resembles a simplified version of the construction of the finger lists of a point, with the difference that we only need to keep track of one point instead of $\alpha$ points.
We maintain the closest point that has been seen so far (line~\ref{line:closest}) and output it.

\begin{algorithm}[t]
\caption{\fname{NearestNeighbor}($q$, $S$, $\listoflist{}$) from~\cite{karger2002finding}\label{alg:orig-find}}
\small
\DontPrintSemicolon
\SetKw{Break}{break}

\KwIn{A list of randomly permuted points $S=s_{1 \ldots n}$, the query point $q$, and the \msl{} $\listoflist{}$.}
\KwOut{The closest point to $q$ in $S$.}
\smallskip

\( \cur \gets 1, m \gets s_1 \)\\

\While{true\label{line:seq-loop}}{
    \( r \gets 2d(s_{\cur}, q) \)\label{line:set-r}\\
    $F^*\gets\findlistradius{\cur}{r}$\\
    \If{$\exists s_j \in F^*$ such that \( d(s_j, q) < d(s_{\cur}, q) \)}{
        $\nxt \gets$ the first $j$ with this property in \( F^* \)\\
        \lIf{\( d(s_{\nxt}, q) < d(m, q) \)}{
            \( m \gets s_{\nxt} \)\label{line:closest}
        }
    }
    \Else{
    \lIf{$|F^*|<\alpha$}{
        \Break\label{line:find-break}
    }
    \lElse{
        $\nxt \gets$ the last index in $F^*$\label{line:find-else}
    }
    }
    \( \cur \gets \nxt \)
}
\Return $m$

\end{algorithm} 

We can also adapt \cref{alg:orig-find} for other queries.  
For $k$-nearest neighbor search, we will keep a set of $k$ points, which is very similar to the construction algorithm that maintains $\alpha$ points.
For range query with distance $\delta$, we can set the search radius \( r \gets d(s_i, q)+\max(\delta,d(s_i, q)) \) on line~\ref{line:set-r}.

\begin{theorem}[\cite{karger2002finding}]
    Given a \msl{} maintaining points with a constant expansion rate, nearest neighbor and $k$-nearest neighbor can be computed in $O(\log{n}\log \log n)$ and $O(k\log{n}\log \log n)$ work \whp{}, respectively; range query can be answered in $O((k'+\log n)\log \log n)$ work \whp{} where $k'$ is output size.    
\end{theorem}

\hide{
\begin{figure*}
  \centering
  \includegraphics[width=\textwidth]{figure/finger_lists.pdf}
  \caption{A visualization of finger lists for different radii.}\label{fig:finger_lists}
\end{figure*}

\begin{figure*}
  \centering
  \includegraphics[width=\textwidth]{figure/finger_lists_2.pdf}
  \caption{A visualization of finger lists for different radii.}\label{fig:finger_lists_2}
\end{figure*}}

\hide{~

The \msl{s} are designed for the nearest neighbor search. 
The design is motivated by an interesting idea of the random-walk procedure to reach the nearest neighbor of a query point $q$ among an input set $S$, using \emph{sampling}. 
At a high level, 
one can start from any point $p\in S$ and attempt to 
iteratively walk to an (ideally closer) point. 
Due to triangle inequality, if there exists a closer point than $p$ to $q$, it must be within \( B_S(p, 2d(p, q))\). 
Therefore, the \emph{sampling scheme} checks a random sample of $\alpha'$ points in \( B_S(p, 2d(p, q))\), 
and walks to the one closest to $q$ as the next candidate point and repeats. 
This process is given in \cref{alg:intuition}. 
With properly set sample rate, each step has a sufficient probability to ``make progress'' by finding a ``much closer'' point, despite a small probability to find a farther point.
Existing work has shown that the following lemma holds by setting $\alpha'=2c^3$ for expansion rate $c$.

\begin{compactlem}[\cite{karger2002finding,ruhl2003efficient}]\label{lem:intuition}
  Given a metric space $(S,d_S)$ with expansion rate $c$, the random walk in \cref{alg:intuition} will reach the nearest neighbor of $q$ in $O(\log |S|)$ iterations \whp{}.
\end{compactlem}

\begin{algorithm}[t]
\small
\caption{The random-walk process to $q$'s nearest neighbor based on sampling~\cite{karger2002finding,ruhl2003efficient}\label{alg:intuition}}
\KwIn{A set of points \( S \) and a query point \( q \)}
\DontPrintSemicolon
Let \( p \) be an arbitrary point in \( S \)\\
\Repeat{\( p \) is the nearest neighbor of \( q \) in \( S \)}{
    Let \( S' \) be a random sample of \( 2c^3 \) points in \( B_S(p, 2d(p, q)) \)\label{line:sample}\\
    Update $p$ as the closest point to \( q \) in \( S'\)\\
}

\end{algorithm}

While a formal proof of this lemma is given in~\cite{karger2002finding,ruhl2003efficient}, 
for the ease of reading, here we provide a proof sketch, 
which will also be useful to understand our new algorithm and its analysis. 
To understand this random-walk procedure, we use the concept of \defn{rank} (w.r.t. the query point $q$) for each point $p\in S$. 
We say a point $p$ has rank $k$ when $p$ is the $k$-th closest point to $q$ among all points in $S$. 
We will show that during the random walk, the rank of $p$ will quickly decrease, 
and finally drops to $1$ in $O(\log n)$ iterations \whp{}. 

Let the current point being processed be $p$ with rank $k$. 
We will use $\delta=d(p,q)$ and simply use $B$ to represent the ball $B_S(p, 2\delta)$ of interest.
We first prove $|B| \le c^2 k$. 
Consider all points and their distance to $q$. 
Since $p$ has rank $k$, there are at most $k$ points in the ball $B_S(q,\delta)$. 
For another point $p'$ in $B$, we have $d(q,p')\le d(q,p)+d(p,p')\le \delta +2\delta\le 4\delta$. 
Therefore, all points in $B$ must also belong to $B_S(q,4\delta)$. 
Considering the definition of expansion rate and $|B_S(q,\delta)|\le k$, we have $|B|\le |B_S(q,4\delta)| \le c^2 k$. 

Therefore, the next point is selected by sampling $\alpha'=2c^3$ points among at most $c^2 k$ points.
We will show that in this case, the closest point is likely to have a rank smaller than $k/2$. 
Consider the opposite case where all $2c^3$ samples are beyond rank $k/2$, this happens with probability at most 
$\left(1-\frac{k/2}{c^2k}\right)^{2c^3} =\left(1-\frac{1}{2c^2}\right)^{2c^3}\le e^{-c}$. 
Therefore, the rank of the next candidate point will be halved with a large probability $1-e^{-c}$. 
In the other case where this does not happen, the candidate point will still be in ball $B$, which has rank at most $c^2k$. 
This means that even in the worst case, the rank may increase by at most a factor of $c^2$, and this happens with a low probability $e^{-c}$. 
Combining the two cases, the rank will still shrink by a constant factor in expectation, and the rank-one point (the nearest neighbor of $q$) will be reached in
$O(\log n)$ iterations \whp{}. 

The nearest-neighbor search process shown in \cref{alg:intuition} is based on a carefully designed sampling scheme.
As such, the search algorithm will iteratively visit a sequence of points ($p_i$) that gets successively closer to the query $q$.
In particular, the next point $p_{i+1}$ is the closest point to $q$ in a random sample of \( (4 + 2 \log \log c)c^2 \) points in \( B_S(p_i, 2d(p_i, q))\).
We will later show how this sample set suffices high probability for the success rate of a query.
This sequence resembles a random walk process similar to the skiplist, so this data structure is referred to as the \msl{}.

Let $r_i$ be the $q$-rank of $p_i$ in $S$.
By applying the triangle inequality and the doubling property of the metric space,
we bound the size of the sampling ball $B$ relative to the current rank $r_i$, specifically $|B| < c^2 r_i$.
Thus, when the sampled size is sufficiently large, 
there is a high probability $P(\alpha)$ that the $q$-rank of $p_{i+1}$ has a $q$-rank less than $r_i/2$.
On the other hand, with probability $1-P(\alpha)$, we cannot halve the $q$-rank,
but we still have $r_{i+1}\le c^2 r_i$.
When $\alpha=(4+2\log\log{c})c^2$, then $\log{r_i}$ is expected to decrease by at least a constant in each iteration.
Finally, by analyzing the expected change $E[\log r_{i+1}]$,
we show that $\log r_i$ decreases by a constant in each step, which, via a Chernoff bound argument,
guarantees convergence to the nearest neighbor ($r=1$) within $O(\log n)$ steps with high probability.

~

~

\cref{alg:intuition} only intuitively explains how the sampling help find the nearest neighbor.
It does not tell how to find the sample set, as well as how to check if $p_i$ is the nearest neighbor of $q$ in $S$.

The key idea in realizing this algorithm is the concept of the finger list, to pre-choose the samples for each point in $S$.
To get start, we first randomly permute all records in the input since that will facilitate the sampling process.
Then, we define the \defn{finger list} $F(i,r)$ of $s_i$ with radius $r$ to contain the first $\alpha=(19+6\log \log c)c^2$ indices after $i$ such that $d(s_i, s_j) < r$ for all $j \in F(i,r)$. 
The finger lists will effectively be the sample set used in \cref{alg:intuition}, despite with a larger size due to the lost of randomness in this given approach to generate them, instead of drawing an independent set at each time.
We will first show that the total number of figure lists is bounded.

\begin{lemma}
  Let each finger list $F(i,r)$ contain $\alpha$ elements. Then whp $|\cup_{r>0} F(i,r)| = O(\alpha\log{n})$.
\end{lemma}

The proof can be found in~\cite{karger2002finding,ruhl2003efficient}.

~

The construction of the \msl is based on the query algorithm, and build the finger lists of all elements iteratively, from the end of the list to the front.
The finger lists are generated during the query process, which will be recorded and stored.
Note that the main difference here is that, here in the construction, we are doing a $\alpha$-nearest neighbor query for some $\alpha>1$.
Hence, we need to modify the query algorithm to adapt to this.
In \cref{alg:orig-build}, we give the pseudocode of an $O(n\log n\log \log n)$ algorithm.

The main difference in the search process in \textsc{Build-Finger} and \textsc{Find} is that instead of maintaining the current ``closest'' point, we maintain a list of them in $F$ of size $\alpha$.

\begin{lemma}
  The work for \cref{alg:orig-find} is $O(\log n\log \log n)$ \whp{}.
\end{lemma}

\begin{proof}[Proof Sketch]
  We refer the reader to \cite{ruhl2003efficient} to see the analysis of random-walk length (number of iterations on \cref{line:seq-loop}), which is also $O(\log n)$ \whp{}.  The intuition is very similar to what was described for \cref{lem:intuition}, but with slightly different constants.
  Another difference is that we need to access $F(i,r)$ throughout the algorithms.
  Although we know there are $O(\log n)$ finger lists \whp{} associated to the $i$-th record, we need a binary search on $r$ to locate the exact finger list. 
  This introduce the additional $\log \log n$ in the cost bound.
\end{proof}

We will later show how to improve the work to $O(\log n)$ \whp{} in \cref{sec:work-efficient}.

}

\section{A More Efficient Sequential Algorithm}\label{sec:work-efficient}

In this section, we will explain our improved sequential algorithms for nearest-neighbor queries and the construction of the \msl{s}.
As mentioned, the main overhead of the original algorithm comes from looking up $\findlistradius{i}{r}$ using binary search. 
By replacing $\findlistradius{i}{r}$ with two algorithmic primitives: an algorithm $\algn{\findlistrank{i}{k}}{r}$ and 
an array of augmented pointers $\rght{\findlistrank{i}{k}}[j]$ for each finger list, 
we can achieve $O(n\log n)$ expected work for construction and $O(\log n)$ expected work for nearest-neighbor query, under the same assumptions. 
We describe the definitions of $\algntxt{}$ and $\rghttxt$ in \cref{tab:notation}, and introduce more details below. 

Since our algorithms are considerably more complicated than the original ones, we will start with the query algorithm in \cref{sec:seq-query} and its analysis.
We will then introduce the construction algorithm in \cref{sec:seq-constr} as well as an amortized analysis for its cost bounds.

\subsection{The Query Algorithms}\label{sec:seq-query}

\begin{algorithm}[t]
\caption{\fname{NearestNeighbor}($q$, $S$, $\listoflist{}$) in $O(\log{n})$ cost\label{alg:seq-find}}
\small
\DontPrintSemicolon
\SetAlgoLined
\SetKw{Break}{break}

\KwIn{A list of randomly permuted points $S=s_{1 \ldots n}$, the query point $q$, and the \msl{} $\listoflist{}$.}
\KwOut{The closest point to $q$ in $S$.}
\medskip

$\cur \gets 1, k\gets 1 $\\
$m \gets s_1$ \tcp*[r]{the closest point to $q$ found so far}

\While{true\label{line:seq-query-main}}{
    $r \gets 2d(s_\cur, q)$\\
    $F^* \gets \algn{\findlistrank{\cur}{k}}{r}$\label{line:seq-align} \tcp*[f]{aligns to the first list with radius $\le r$}\\
    
    \If{$\exists s_j \in F^*$ such that \( d(s_j, q) < d(s_\cur, q) \)}{
        $\nxt \gets$ the first $j$ with this property in \( F^* \)\\
        \lIf{$d(q, s_\nxt) < d(q, m)$}{
            $m \gets s_\nxt$
        }
    }
    \Else{
    \lIf{$\left|F^*\right|<\alpha$}{
        \Break
    }
    \lElse{
        $\nxt \gets$ the last index in \( F^* \) \label{line:seq-last}
    }
    }
    
    $k \gets \rght{F^*}[\nxt]$, $\cur\gets \nxt$
}

\medskip

\tcp{This function re-aligns the current finger list to the first list with radius $r'\le r$ by traversing up or down in $\findlistrank{i}{\cdot}$}
\myfunc(){$\algn{\findlistrank{i}{k}}{r}$\label{line:align-begin}}{
    $k'\gets k$\\
    \tcp{Based on whether $r$ is larger or smaller than $\radius{\findlistrank{i}{k}}$, only one of the two while-loops below will be executed.}
    \lWhile(\tcp*[f]{up}){$k'>1$ and $\radius{\findlistrank{i}{k'-1}} \le r$}{
        $k' \gets k'-1$\label{line:up}
    }
    \lWhile(\tcp*[f]{down}){$\radius{\findlistrank{i}{k'}} > r$}{
        $k' \gets k'+1$\label{line:align-end}
    }
    \Return $\findlistrank{i}{k'}$\\
}

\end{algorithm}

The key idea in our query algorithm, as shown in \cref{alg:seq-find}, 
is to use the \emph{\algntxt} algorithm and the \emph{.\rghttxt} pointers to accelerate the query process.
$\algn{\findlistrank{\cur}{k}}{r}$ is used to calibrate the index $k$ to $k'$ in $\findlistrank{\cur}{\cdot}$, such that $\findlistrank{\cur}{k'}$ is $\findlistradius{\cur}{r}$.
This can be implemented easily by traversing $\findlistrank{\cur}{\cdot}$ up or down from index $k$, since all finger lists are sorted by their radii. 
The process is shown on line~\ref{line:align-begin}--\ref{line:align-end}. 
We will show that whenever $\algn{\findlistrank{\cur}{k}}{r}$ is called, the destination is likely not too far from $k$, 
which bounds the total cost caused by \emph{\algntxt} during the query algorithm (\cref{lem:total-align}). 
Meanwhile, here we augment each finger list $\findlistrank{i}{k}$ with an array $\rghttxt[j]$, which maps each $j\in \findlistrank{i}{k}$ to $\findlistradius{j}{r}$, where $r=\radius{\findlistrank{i}{k}}$ is the radius of the current list. 
This allows us to advance to the next \focuslist{} quickly without a binary search---when we shift the focus from $\cur$ to $\nxt$, 
we first advance to the finger list in $\nxt$ with the same radius using the $.\rghttxt$ pointer. 
We then calibrate around this list by moving up or down to determine the actual \focuslist{}. 
We will later show how the $.\rghttxt$ pointers are built efficiently in construction in \cref{sec:seq-constr}.
Here we assume it is correctly computed, so we can introduce our query algorithm.
Due to space limit, we only provide proof sketches here and defer the full proofs in \ifconference{the full version \cite{metricskiplistfullversion}}\iffullversion{Appendix.~\ref{sec:app-proof}}.

\begin{lemma}\label{lem:4.1}
Given $\listoflist{i}$, there can be $O(1)$ expected lists between $\listoflist{i}(r)$ and $\listoflist{i}(2r)$.
\end{lemma}

\begin{proof}[Proof sketch]
Due to the low expansion rate, $|B(s_i,2r)|/|B(s_i,r)|\le c$.
Consider the points in $B(s_i,2r)\setminus B(s_i,r)$ in their random permutation order.
A point can only be added to $\listoflist{i}$ (thus created a new finger list) before we have seen $\alpha$ points from $B(s_i,r)$ in this order.
Hence, the expected number of such points is $\alpha \cdot|B(s_i,2r)|/|B(s_i,r)|=O(\alpha c)=O(1)$.
\end{proof}

\begin{lemma}\label{lem:4.2}
The total number of iterations executed on line~\ref{line:up} is $O(\log n)$ in expectation.
\end{lemma}

The proof of this lemma is based on \cref{lem:4.1} and is provided in \ifconference{the full version \cite{metricskiplistfullversion}}\iffullversion{Appendix.~\ref{sec:app-proof}} due to the space limit.

\begin{lemma}
\label{lem:total-align}\label{lem:4.3}
The total work in \emph{align} is $O(\log n)$ in expectation.
\end{lemma}

\begin{proof}[Proof Sketch]
\cref{lem:4.2} has bounded the expected work of the upward moves (\(k\gets k-1\), line~\ref{line:up}). Here we only need to bound the cost of downward moves (\(k\gets k+1\), line~\ref{line:align-end}): in total the pointers are moved for \(O(\log n)\) times in expectation.
This can be achieved by defining a potential function of the ball size (or the rank) and analyzing the random walk of the potential function.
We use $\Phi=\log_2|B(s_{\cur},2d(s_{\cur},q))|$.
As discussed in \cref{sec:kr} and \cref{lem:intuition}, after \(O(1)\) new elements in the list, the ball size is expected to halve; equivalently, every \(O(1)\) downward \emph{align} move decreases \(\Phi\) by \(\Theta(1)\) in expectation.
Hence, the number of downward moves is proportional (in expectation) to the total decrease of \(\Phi\).
Note that \(\Phi\) can grow in the upward moves.
Each time \cref{line:up} is executed, the search radius at most doubles, so \(\Phi\) can grow by at most a factor of $c=O(1)$.
Hence, we need the same expected asymptotic number of downward moves to account for the decrease of \(\Phi\).
Putting all pieces together gives $O(\log n)$ expected number of downward moves, and the same cost for the \emph{align} function.
\end{proof}

\begin{theorem}
    Given a \msl{} maintaining points with a constant expansion rate, \cref{alg:seq-find} computes the nearest neighbor in $O(\log{n})$ expected work.
\end{theorem}

\begin{proof}
\cref{alg:seq-find} follows the same logic as \cref{alg:orig-find}. The difference is that $\findlistradius{\cur}{r}$ is now obtained by using the \emph{\rghttxt{}} pointers and local adjustment by the \emph{align} function, instead of binary searches in $O(\log \log n)$ cost.
\cref{lem:4.3} proves that the total cost in \emph{align} is $O(\log n)$ in expectation.
The \emph{\rghttxt{}} pointers are chased once in each iteration in the random walk, so the total cost here is $O(\log n)$ expected.
Combining these costs proves the theorem.
\end{proof}

We can change the candidate set from one point to $k$ points, which yields $O(k\log{n})$ expected work for $k$-nearest neighbor query.

\subsection{The Sequential Construction Algorithm}\label{sec:seq-constr}

With the insights from the query algorithm, we now explain the construction algorithm that will generate the $\rghttxt$ pointers efficiently.
We provide our algorithm in \cref{alg:seq-efficient-build}.

The overall idea in \cref{alg:seq-efficient-build} is quite similar to \cref{alg:seq-find}, except that we will now keep $\alpha$ points rather than one point.
Hence, in each iteration of the random walk, we first compute the current search radius $r$,
then we use $\algn{\findlistrank{\cur}{k}}{r}$ to find the finger list $\findlistrank{\cur}{k'}$ (effectively $\findlistradius{\cur}{r}$ in \cref{alg:orig-build}).
Then we perform the same check conditions, find the next focus point $\nxt$, and update $\listoflist{i}$ if needed.
At the end of the iteration, we take the $\rght{F^*}[\nxt]$ pointer to the next finger list we focus on, just like in \cref{alg:seq-find}.
As the construction algorithm, we additionally need to build the \emph{\rghttxt{}} pointers: on lines \ref{line:alg:seq-efficient-build:initial-r-start}--\ref{line:alg:seq-efficient-build:initial-r-end}
and \ref{line:alg:seq-efficient-build:new-r-start}--\ref{line:alg:seq-efficient-build:new-r-end}.

Recall that the definition is $\rght{\findlistradius{i}{r}}[j]=\findlistradius{j}{r}$.
In the initial finger list $\findlistrank{i}{1}$,
we first set $\rght{\findlistrank{i}{1}}[j]=\findlistrank{j}{1}$ for all $j\in \findlistrank{i}{1}$.
However, this finger list can have a higher radius than $\findlistrank{i}{1}$,
so we may need to decrease the radius of $\rght{\findlistrank{i}{1}}[j]$ until it is at most $\radius{\findlistrank{i}{1}}$.
Here we call a function \fname{Align-All}($F$),
which checks all $j$ in $F$ and decreases the radius of $\rght{F}[j]$ until it is at most $\radius{F}$.

When $s_i$ gets a new finger list $F_{\mathit{new}}$,
we build the $\rght{F_{\mathit{new}}}[\cdot]$ pointers for the points in it.
Note that $F_{\mathit{new}}$ is created by replacing the farthest point in the focus list $F$ with $s_\nxt$.
For all other points $s_j$ in $F_{\mathit{new}}$, their \emph{\rghttxt{}} pointers remain unchanged,
i.e. $\rght{F_{\mathit{new}}}[j]=\rght{F}[j]$.
Then, for the evictor $s_\nxt$, we set $\rght{F_{\mathit{new}}}[\nxt]$ to $\rght{F^*}[\nxt]$.
Then we call \fname{Align-All}($F_{\mathit{new}}$) to align the radius of all $\rght{F_{\mathit{new}}}[\cdot]$ pointers if needed.
\begin{algorithm}[t]
\DontPrintSemicolon
\small
\caption{Sequential $O(n\log{n})$ Construction Algorithm\label{alg:seq-efficient-build}}
\SetKw{Break}{break}

\SetKwInput{Maintains}{Maintains}

\For{$i \gets n$ \KwTo $1$}{
    \fname{Build-Finger}($i$)
}

\smallskip

\myfunc(){\upshape{\fname{Build-Finger}}($i$)}{

$\listcontent{\listoflist{i}[1]}\gets \{s_{i+1}, s_{i+2}, \dots, s_{\min\{i+\alpha,n\}}\}$\label{line:alg:krconstruct:first}\\
\ForEach{point $s_j \in \findlistrank{i}{1}$\label{line:alg:seq-efficient-build:initial-r-start}}{
    $\rght{\findlistrank{i}{1}}[j] \gets \findlistrank{j}{1}$
}
\fname{Align-All}($\findlistrank{i}{1}$)\label{line:alg:seq-efficient-build:initial-r-end}\\
$\cur \gets i + \alpha$ \\
$k \gets 1$\tcp*[f]{Start from the largest finger list of $s_\cur$}\\

\While{true}{
    $F \gets$ the last finger list of $s_i$\\
    $r \gets d(s_{\cur}, s_{i}) + \max(\radius{F}, d(s_{\cur}, s_i))$\\
    $F^* \gets \algn{\findlistrank{\cur}{k}}{r}$\label{line:adjust-fw}\\

    \If{$\exists s_j\in F^*$ such that $d(s_j, s_i) < \max(\radius{F}, d(s_{\cur}, s_i))$}{
        $\nxt \gets$ the first $j$ with this property in $F^*$\\
        \If{$d(s_\nxt, s_i) < \radius{F}$}{
            $F_{\mathit{new}}\gets $ replacing the farthest point in $F$ with $s_{\nxt}$\\
            $\rght{F_{\mathit{new}}}[\nxt] \gets \rght{F^*}[\nxt]$\label{line:alg:seq-efficient-build:new-r-start}\\
            \fname{Align-All}($F_{\mathit{new}}$)\label{line:alg:seq-efficient-build:new-r-end}\\
            Append $F_{\mathit{new}}$ to $\listoflist{i}$\\
        }
    }
    \Else{
    \lIf{$\left|F^*\right|<\alpha$}{
        \Break
    }
    \lElse{
        $\nxt \gets$ the last index in $F^*$
    }
    }
    $k \gets \rght{F^*}[\nxt]$, $\cur\gets \nxt$\label{line:fw-gets-r-fw-sj}\\
}

}

\smallskip

\myfunc(){\upshape{\fname{Align-All}}($F$)}{
    \ForEach{point $s_j \in F$}{
        $\algn{\rght{F}[j]}{\radius{F}}$\\
    }
}

\end{algorithm}

\begin{theorem}\label{thm:seq-build}
Given a randomly permuted point set $S=\{s_{1..n}\}$ with a constant expansion rate and a constant $\alpha$, \cref{alg:seq-efficient-build} builds the \msl{} in $O(n\log{n})$ expected work.
\end{theorem}

We first show a useful lemma.

\begin{lemma}\label{lem:radius-halve}
    The radius of $F$ is expected to halve after adding a constant number of new elements.
\end{lemma}

\begin{proof}
    The proof of Alg 5.18 in \cite{ruhl2003efficient} shows that, $B(i,\radius{\listoflist{j+2\alpha}})$ is expected to halve compared to $B(i,\radius{\listoflist{j}})$.
    Hence, due to the expansion rate, the radius of $F$ is expected to halve after seeing $2\alpha\lceil\log{c}\rceil$ new evictors (repeated $\lceil\log{c}\rceil$ times).
\end{proof}

\medskip
\begin{proof}[Proof of \cref{thm:seq-build}]
We show that the work for \fname{Build-Finger} is $O(\log n)$ in expectation in each invocation, which implies the theorem.
The analysis of \cref{alg:orig-build} shows that the random walk has length $O(\log n)$ \whp{}, and the proof of \cref{alg:seq-find} shows that the \emph{\algntxt} has $O(\log n)$ expected work.
Hence, the only leftover part is the work spent on building the \emph{\rghttxt{}} pointers.

Recall that the \emph{\rghttxt{}} pointer is defined as $\rght{\findlistradius{i}{r}}[j]=\findlistradius{j}{r}$.
We now show that for each $s_j$,
the cost in building all $\rght{\findlistrank{i}{\cdot}}[j]$ pointers
can be charged to other operations in \fname{Build-Finger}$(i)$ asymptotically.
Hence, the total work is  $O(\log n)$ expected per random walk, and $O(n \log n)$ expected in total.

For each $s_\nxt$ added to some finger lists of $s_{i}$, maintaining its \emph{\rghttxt{}} pointers appears in two phases.
\begin{itemize}
    \item Phase 1: the first appearance.
    When $s_\nxt$ is first added to some finger list $F_{\mathit{new}}$, we set $\rght{F_{\mathit{new}}}[\nxt]=\rght{F^*}[\nxt]$ (line~\ref{line:alg:seq-efficient-build:new-r-start}).
    Then we align the radius of $\rght{F_{\mathit{new}}}[\nxt]$ until its radius is at most $\radius{F_{\mathit{new}}}$.
    \item Phase 2: the follow-up copies.
    When $s_{i}$'s finger lists get updated by $F_{\mathit{new}}$, for all points $s_j$ not evicted, 
    we copy $\rght{F_{\mathit{new}}}[j]=\rght{F}[j]$.
    Then we keep decreasing the radius of $\rght{F_{\mathit{new}}}[j]$ until its radius is at most $\radius{F_{\mathit{new}}}$.
\end{itemize}

\myparagraph{Phase 1.}
We first show that the running time of Phase 1
can be charged to the running time of line~\ref{line:adjust-fw} in the next iteration of the while loop.
For $\rght{F_{\mathit{new}}}[\nxt]$, it is initialized as $\rght{F^*}[\nxt]$,
then it is aligned to have radius at most $\radius{F_{\mathit{new}}}$.
For its counterpart, in line~\ref{line:fw-gets-r-fw-sj}, $F^*$ is also updated to $\rght{F^*}[\nxt]$,
then in the next iteration, it is adjusted to have radius at most $r=d(s_{\nxt}, s_i)+\max(\radius{F_{\mathit{new}}}, d(s_{\nxt},s_i))=d(s_{\nxt},s_i)+\radius{F_{\mathit{new}}}$.
This is because $s_{\nxt}$ will become the $s_{\cur}$, and $F_{\mathit{new}}$ will become the $F$ in the next iteration,
and we also have $d(s_{\nxt},s_i)<\radius{F_{\mathit{new}}}$ because $s_{\nxt}$ is inside $F_{\mathit{new}}$.

Note that $d(s_{\nxt},s_i)+\radius{F_{\mathit{new}}} < 2\radius{F_{\mathit{new}}}$.
Based on \cref{lem:radius-halve}, in addition to the running time of \cref{line:adjust-fw},
we only need to perform constant expected number of decreases on $\rght{F_{\mathit{new}}}[\nxt]$ to make its radius at most $\radius{F_{\mathit{new}}}$.

\myparagraph{Phase 2.}
A point $s_j$ may show up in several finger lists in $\listoflist{i}$, and we need to update $\rght{\findlistrank{i}{\cdot}}[j]$ in $\listoflist{j}$ by moving down the corresponding \emph{\rghttxt} pointers.
Here we can charge this cost to copying the finger lists, since we need to pay $O(1)$ cost every time we duplicate $s_j$ in finger lists.
Since $s_j$ is in $\findlistrank{i}{\cdot}$, the radii of these finger lists must be at least $d(s_i,s_j)$.
Due to the constant expansion rate, the numbers of points centered at $j$ and $i$ can only differ by a constant factor of $c$.
Hence, the expected number of pointer moving is asymptotically the same as the number of appearances of $s_j$ in $\listoflist{i}$, using the same argument as in the proof of \cref{lem:radius-halve}.

\smallskip 
Combining the two cases proves the theorem.
\end{proof}

\section{Parallel Construction of the \Msl{s}}\label{sec:parallel}

In \cref{sec:work-efficient}, we presented an improved algorithm that constructs the \msl{} in $O(n\log n)$ expected work, and answers queries in $O(\log n)$ expected work. 
In this section, we further present a parallel algorithm for the construction process. 
We note that this parallelization is inherently challenging due to the sequential nature of the process, where building each finger list relies on later finger lists that have already been constructed. 
More specifically, in both our algorithm and the original KR algorithm~\cite{karger2002finding}, building each finger list requires to shift the focus to simulate the random walk process. 
At the stage where the algorithm is focusing on $\cur$, to find the next \focuspoint{} $\nxt$ from the \focuslist{} $F^*$,
$\cur$ not only relies on (or needs to wait for) $\nxt$, but also other points in $F^*$. 
This is because $\nxt$ is chosen as the first or last index in $F^*$ satisfying certain conditions (lines~\ref{line:first},\ref{line:farther} in \cref{alg:orig-build}), 
and thus we may need to know the full list to make the decision. 
This weaves an intricate dependence structure, which itself is also dictated by randomness. 

\hide{
In fact, in both our algorithm and the original KR algorithm~\cite{karger2002finding}, an $O(\alpha n\log n)$ number of finger lists needs to be built, and building it would require another $O(\alpha)$ finger lists each corresponding to one point the current finger list, which are used to decide the next point in the random-walk procedure.
This weaves a complicated dependence structure (the dependence DAG) containing $O(\alpha n\log n)$ vertices and $O(\alpha^2 n\log n)$ edges, which itself is also determined by randomness.
Analyzing this DAG and showing high parallelism of this computation can be extremely challenging, or even unlikely in our perspective.
}

To overcome this challenge and parallelize the algorithm, our solution is somewhat counterintuitive.
We add artificial synchronization barriers to the dependence graph---in a divide-and-conquer manner.
We will show that the dependences crossing the synchronization barriers actually form a tree structure, and it is possible to bound the tree depth in $O(\log n)$ \whp{}, which enables a polylogarithmic span for the entire algorithm.
However, given the complication of even the sequential algorithms, we will present our parallel algorithms in three steps to improve readability. 
First, in \cref{sec:par-naive}, we describe a simple divide-and-conquer algorithm for the construction of \msl{s}.
This algorithm is not work-efficient, 
but it provides useful intuition to achieve parallelism. 
This algorithm can achieve polylogarithmic span with an $O(\log n\log \log n)$ factor of overhead in work, or slightly better work with $O(n^{\epsilon})$ span.  
Next, we will show the high-level idea on parallelizing the construction algorithm by strictly following the random-walk procedure without doing extra work. 
This algorithm is given in \cref{sec:parallel-optimized}. 
The key challenge and algorithmic insight is analyzing the dependence structure, which we refer to as the \defn{control tree}.
In \cref{sec:control-analysis} we prove the $O(\log{n})$ height bound of the control tree,
then we show how to construct and use it in \cref{sec:par-constr}. 
Finally, we will show how to achieve the work-efficiency by efficiently maintaining the 
$\rghttxt$ pointers for each finger list in each step in \cref{sec:map}.

\subsection{A Simple Work-Inefficient Solution}\label{sec:par-naive}

We start with a simple work-inefficient divide-and-conquer parallel construction algorithm. 
The idea is to break the input sequence into two equal-size chunks, $s_{1.. n/2}$ and $s_{(n/2+1)..n}$, and build their \msl{s} independently.
When they both finish, note that all finger lists on the right half are fully ready.
However, the points in the left half can only see the points in $s_{1.. n/2}$, 
and their finger lists need to further incorporate the points in the right half. 
To complete the construction for the points on the left, we first propose a simple method as follows.
For each point $s_i$ in $s_{1..n/2}$, we run \fname{Build-Finger}$(s_i)$ on the \msl{} for $s_{(n/2+1)..n}$.
This will generate a list of new finger lists for each $s_i$ (denote it as $\bar{\listoflist{i}}$), ``as if'' $s_i$ is at the location of $s_{n/2}$.
This step, using \cref{alg:orig-build}, costs $O(n\log n\log \log n)$ work and $O(\log n\log \log n)$ span \whp{}. 
As such, we would need to merge the two finger lists $\listoflist{i}$ and $\bar{\listoflist{i}}$ for each point,
where we only preserve a subset of points in $\bar{\listoflist{i}}$.

Since both $\listoflist{i}$ and $\bar{\listoflist{i}}$ have $O(\log n)$ finger lists \whp{} and each finger list has a constant size of $\alpha$, combining them will take $O(\log n)$ work \whp{}, leading to:
\begin{align*}
  W(n)&=2W(n/2)+O(n\log n\log \log n) \\
  D(n)&=D(n/2)+O(\log n\log \log n)
\end{align*}
The recurrences yield $O(n \log^2{n}\log{\log{n}})$ work and $O(\log^2 n\log\log n)$ span, both \whp{}. 
Note that in addition to a 2-way divide-and-conquer, we can increase the branching factor to $k$, and when combining, 
we consider merging $k-1$ times, starting from the very last two subproblems to the front, each as mentioned above. 
Now the recurrences become:
\begin{align*}
  W(n)&=k\cdot W(n/k)+O(n\log n\log \log n) \\
  D(n)&=D(n/k)+(k-1)\cdot O(\log n\log \log n)
\end{align*}
The solutions to the recurrences imply $O(n \log{n}\log_k{n}\log{\log{n}})$ work and $O(k\log n\log_k{n}\log\log n)$ span, both \whp{}.
By plugging in $k=n^\epsilon$ for some constant $\epsilon>0$, we can achieve $O(n \log{n}\log{\log{n}})$ work but only a polynomial span.

\subsection{Overview of the Work-Efficient Algorithm}\label{sec:parallel-optimized}

\begin{figure*}
  \centering
  \includegraphics[width=\textwidth]{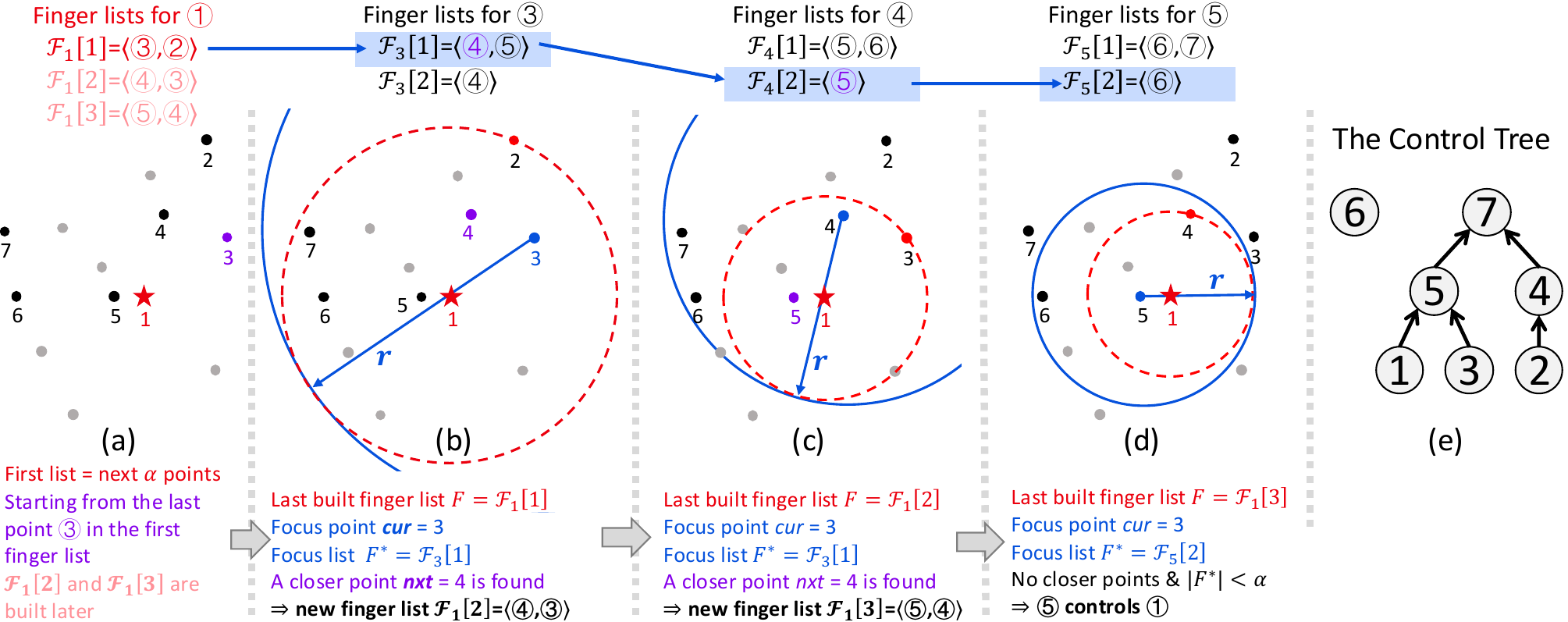}
  \caption{\fullbf{Illustration of parallel construction when processing point $i=1$, and the \control{} tree.}  
  Points 1--7 are in the current (recursive) subproblem.  
  (a)--(d) also illustrate the sequential construction of finger lists for $s_1$,
  showing the random walk steps from it. 
  The point under construction ($i=1$) and its finger lists are marked red. 
  The \focuslist{} ($F^*$) of each step are in blue background. 
  The \focuspoint{} $\cur$ at each step are marked as blue dots. 
  Red dashed circles are the ranges of the last built finger list ($\radius{F}$), centered with point $i=1$.  
  Blue solid circles are the ranges of the \focuslist{} ($\radius{F}$), centered with the focus point $\cur$.  
  The next point $\nxt$ selected are marked purple. 
  Arrows at the top show the random walk among the finger lists.  Figure~(e) is the control tree.
  }\label{fig:control-example}
  \vspace{-0.5em}
\end{figure*}

From \cref{sec:par-naive}, we can see that a straight-forward divide-and-conquer algorithm inevitably leads to redundant computation. 
This motivates us to strictly follow the computation in the random-walk algorithm shown in \cref{alg:seq-efficient-build} to maintain work-efficiency, 
while still leveraging a 2-way divide-and-conquer strategy to achieve high parallelism. 
The key remaining question is whether in this case, the two subproblems can be \emph{merged} in polylogarithmic span; 
with $O(\log n)$ recursion depth, this would imply an overall polylogarithmic span. 
We demonstrate that this is indeed achievable by giving an algorithm in \cref{alg:par-build}, 
which finally leads to an $O(n\log n)$ work-efficient construction algorithm with polylogarithmic span. 

\begin{theorem}\label{thm:main}
Given a randomly permuted point set $S=\{s_{1..n}\}$ with a constant expansion rate and a constant $\alpha$, the \msl{} can be built in $O(n\log{n})$ expected work and $O(\log^3 n)$ span \whp{}.
\end{theorem}

We first overview the main challenge and our high-level idea. 
As mentioned in \cref{sec:par-naive}, the complication in the parallel divide-and-conquer algorithm stems from the fact that the points of the first (left) half are unaware of points in the second (right) half. 
After both recursive calls finish, the left-half finger lists are incomplete and must further incorporate the points on the right. 
This necessitates a \fname{Merge} scheme to complete the finger lists on the left. 
The goal is to finish the ``list completion'' task for all left-half points \emph{in parallel} and \emph{work-efficiently}. 
In \cref{sec:par-naive}, the \fname{Merge} function in the \naive{} algorithm achieves parallelism by letting each point on the left restart a random walk from the first point on the right, i.e., $s_{n/2+1}$. 
As such, all random walks are independent with each other. 
However, this is work-inefficient since the computation differs from the sequential algorithm. 
During the process, more finger lists may be created and then merged/discarded, causing extra work. 

Intuitively, an efficient simulation of the sequential process should let each point restart the random walk from \emph{itself}, 
following the same route as the sequential algorithm. 
Therefore, for some $s_i$ on the left, its starting point on the right should not simply default to $s_{n/2+1}$; 
rather, it should be guided by its last ``hop'' on the left in the original random walk. 
However, when this ``previous hop'', which is another point on the left half, also remains uncompleted, a dependence arises, which hinders parallelism. 

With this in mind, we now formally define the dependence structure among the points. 
Given a randomly permuted sequence of input points $S=\{s_{1..n}\}$, we consider the standard sequential construction of finger lists for all points from $s_{n'}$ to $s_1$.
For each point $s_i$, its corresponding random walk terminates at its last \focuspoint{} $s_{\cur}$ (line \ref{line:alg:krconstruct:notfull} in \cref{alg:orig-build}). 
In this case, we say $s_{\cur}$ \defn{\control{s}} $s_{i}$ (and $s_{\cur}$ is $s_{i}$'s \defn{\control{} point}). 

Now consider $S=\{s_{1..n'}\}$ as a left-side subproblem in the divide-and-conquer scheme. 
In this case, we will need to complete (i.e., \fname{Merge} the points on the right into) the finger lists for each $s_i$ by simulating random walks from $s_i$. 
The random walk should normally advance to its control point $s_{\cur}$, and then, it should be guided by the completed version of $\listoflist{\cur}$ to advance to the right side. 
Clearly, the random walk of $s_i$ cannot determine the next hop $\nxt$ before $s_{\cur}$'s finger lists are fully ready, thereby incurring a dependence. 
In other words, for a point $s_i$, if we consider the intended random walk process in the full list (both left and right halves considered), 
then its \control{} point is exactly the last hop before advancing across the boundary between the two halves. 

We present an illustration of the control point and the dependence caused by it in \cref{fig:control-example}. 
In this example, we have points $s_{1..7}$ as the left-side subproblem in a divide-and-conquer scheme. 
In the subproblem on $s_{1..7}$, the finger lists of $s_1$ are constructed by initiating a random walk, which goes through $s_3,s_4,s_5$.
It terminates at $\findlistrank{5}{2}=\langle 6 \rangle$, since no closer point is found, and this list is incomplete (line~\ref{line:alg:krconstruct:notfull} in \cref{alg:orig-build}). 
By definition, $s_5$ controls $s_1$.
Since $\findlistrank{5}{2}$ is incomplete, we can expect more points to be added with the right half considered. 
For $s_1$, again with the right half considered, its next focus point will be determined by what other points are added to $\findlistrank{5}{2}$:
either a closer point, or, if all points are still too far, the last (largest index) will be selected. 
In both cases, $\nxt$ will need to wait for $\listoflist{5}$ to incorporate information on the right. 
Therefore, the \fname{Merge} algorithm must first settle down $s_5$ before processing $s_1$. 

For a sequence $S=\{s_{1..n'}\}$, we define the \defn{control tree} as a graph (forest), where each point is a vertex, and if $s_i$ is controlled by $s_j$, we add an edge from $i$ to $j$. For a left subproblem on $S$, its control tree captures the dependence of processing all the points in the \fname{Merge} process. 
We illustrate this in \cref{fig:control-example}. 

\myparagraph{A Useful Fact about Control Points.} For a point $s_i$ initially blocked by its \control{} point $s_{\cur}$, 
once $s_{\cur}$ finds its next hop on the right, the fully-built finger lists on the right will allow the random walk to proceed until the end of the right half. 
By definition, this process will stop at another incomplete finger list (line~\ref{line:alg:krconstruct:notfull} in \cref{alg:orig-build}). Interestingly, this will exactly be the control point of $s_i$ in a future recursive call, when the next time $s_i$ (and this new control point) are both in a left subproblem. 
Our algorithm will use this property to avoid recomputing the control points in each recursive call. 

In the remainder of this section, we present more details about the \control{} trees and our algorithm based on it. 
First, we prove that the \control{} tree has the height of $O(\log n)$ \whp{}, which will be detailed in \cref{sec:control-analysis}. 
Note that the \control{} tree is inherent to the \msl{s} and independent with any construction algorithm that generates it.
With this in mind, if the tree structure is known, we can simulate random walks for the points in the tree level by level. 
More precisely, we can first settle all points with no \control{} points (i.e., the roots in the control tree/forest), since they can directly advance to the right.
With these points ready, we then process the next level, so on so forth. 
Since $|\listoflist{i}|$ is bounded by $O(\log n)$ \whp{}, we can bound the span to be polylogarithmic as long as we can build the control tree efficiently.

The next question is how to find the \control{} points efficiently, such that we can construct the \control{} tree and execute random walks accordingly.
We show this algorithm in \cref{alg:par-build} and describe it in \cref{sec:par-constr}. 
Given the complication of this algorithm, we will first assume the following invariant for the \emph{\rghttxt{}} pointers is maintained:
for all completed finger lists $\findlistrank{i}{k}$ and $j\in \findlistrank{i}{k}$,
either $\rght{\findlistrank{i}{k}}[j]$ maps to the correct finger list of $j$,
or $\rght{\findlistrank{i}{k}}[j]$ maps to the last completed finger list of $j$; we do not maintain the \emph{\rghttxt{}} pointers for the tail finger lists.
We will show how the \control{} tree can be constructed efficiently with these pointers in a work-efficient and highly parallel manner.

The final question is how to maintain the $.\rghttxt$ pointers.
Interestingly, this appears to be the most challenging part of our algorithm. 
This is because in parallel, only part of the finger lists in $\listoflist{i}$ are built in the execution of an algorithm.
This means while the next focus point can be decided, its corresponding finger list and thus the \emph{\rghttxt} pointers are not available, unlike in the sequential case.
We will show how to overcome this problem in \cref{sec:map}. 
We primarily address two challenges. 
The first is how to complete the random walks work-efficiently when only partial pointers are available.
The second is how to fill in the incomplete \emph{\rghttxt} pointers work-efficiently when the target finger lists become available. 
To achieve this, we again represent the dependence DAG as a tree structure, but apply the computation lazily. 
Combining all these components proves the main theorem, \cref{thm:main}, of this paper.

\subsection{Analysis of the Control Tree}\label{sec:control-analysis}

As discussed, we define the concepts of the control points and thus the control tree above to capture the dependence structure in the parallel construction process. 
We now show that the tree height is shallow, which indicates high overall parallelism.

\begin{lemma}\label{lem:control-tree-height}
    Given a randomly shuffled sequence of points $S=\{s_{1..n}\}$, its corresponding control tree has the height $O(\log{} n)$ \whp{}.
\end{lemma}

\begin{proof}
    Let the control point of point $s_i$ be denoted by $s_{C(i)}$.
    We now claim that if $C(i)\ne i$, then with at least half probability we have $C(i)\ge (i+n)/2$.
    Namely, $s_{C(i)}$ is more likely to be in the second half of $s_{i..n}$.
    Note that $s_{(n-\alpha+1)..n}$ are not controlled by any point.
    
    Let $s_m$ be the closest point to $s_i$ in $s_{(i+1)..n}$.
    Based on the random walk (see \cref{sec:kr}), $s_m$ must be visited in the random walk during the construction of $\listoflist{i}$, so we must have $C(i)\ge m$.
    Since $S$ is shuffled uniformly at random, with a half chance $m$ is at least $(i+n)/2$, so is $C(i)$.
    In addition, the probability distribution is independent for the transition between $i$ and $C(i)$ and the transition between $C(i)$ and $C(C(i))$---this is due to the same reason as explained in the analysis of the KR algorithm for \msl{s} construction in \cref{sec:kr}.
    
    With the distribution of $C(i)$, we can now bound the control tree depth easily.
    Consider the path starting from a node $i$. 
    Every time we chase $C(i)$, it can be considered as a coin flipping with a half chance cutting the remaining range by half.
    These are the simple Bernoulli trials, so it takes $O(\log n)$ \whp{} coin flips to see $\log n$ successful outcomes, which upper bounds the path lengths from $i$.
    Taking the union bound across all points proves this lemma.
\end{proof}

In our algorithm, we use $C[i]$ to store the control point of $s_i$ at the current recursive level.
This value will be updated as the algorithm proceeds, based on the \emph{useful fact} mentioned in \cref{sec:parallel-optimized}. 

\begin{algorithm}[t]
\DontPrintSemicolon
\small
\caption{\fname{Parallel-Build}($s_{l \ldots r}$): Builds the finger lists $s_{l \ldots r}$ in parallel}\label{alg:par-build}
\SetKw{Break}{break}

\SetKwInput{Maintains}{Maintains}
\Maintains{$C[i]$: the control point for $s_i$\\
}

\lIf(\tcp*[f]{Base case}){$r-l<\alpha$}{
    $C[l]\gets l$; \Return
}
$m \gets (l+r)/2$\;
\parDo{}{
\fname{Parallel-Build}$(s_{l \ldots m})$\\
\fname{Parallel-Build}$(s_{(m+1) \ldots r})$\\
}
\tcp{The \fname{Merge} phase: complete the finger lists on the left half}
Construct the control tree for $i \in [l..m]$ in parallel \\

\ForEach{\upshape layer of the control tree}{
    \parForEach{\upshape point $s_i$ in this layer}{
        \lIf{$i+\alpha\le r$} {
            \fname{BuildFromRight}($i, C[i]$)
        }
    }
}

\medskip

\myfunc(){\upshape{\fname{BuildFromRight}}($i,\cur$)}{

\If(\tcp*[f]{\upshape $i$ has no completed finger list}){$i=\cur$}{
    $\listcontent{\listoflist{i}[1]}\gets \{s_{i+1}, s_{i+2}, \dots, s_{i+\alpha}\}$, $\cur\gets i+\alpha$
}

\While{true}{
    $F \gets$ the last complete finger list of $s_i$\\
    $r \gets d(s_{\cur}, s_{i}) + \max(\radius{F}, d(s_{\cur}, s_i))$\\
    
    $F^*\gets \findlistradius{\cur}{r}$\\
    \If{$\exists s_j\in F^*$ s.t. $d(s_j, s_i) < \max(\radius{F}, d(s_{\cur}, s_i))$}{
        $\nxt \gets$ the first $j$ with this property in $\findlistradius{\cur}{r}$\\
        \If{$d(s_{\nxt}, s_i) < \radius{F}$}{
            $F_{\mathit{new}}\gets $ replacing the farthest point in $F$ with $s_{\nxt}$\\
            Append $F_{\mathit{new}}$ to $\listoflist{i}$\\
        }
    }
    \Else{
    \If{$\left|F^*\right|<\alpha$\label{line:par-control}}{
        $C[i] \gets \cur$\\
        Build the tail and  \Break
    }
    \lElse{
        $\nxt \gets$ the last index in $F^*$
    }
    }
    $\cur\gets \nxt$\\
}

}

\end{algorithm}

\subsection{The Construction Algorithm}\label{sec:par-constr}

\cref{alg:par-build} presents the core ideas of our parallel construction algorithm.
Here, for the ease of understanding, we use $\findlistradius{\cur}{r}$ to acquire the focus list, and later show how to replace it with \emph{.\rghttxt{}} pointers and the \emph{\algntxt{}} function in \cref{sec:map}.

Our \fname{Parallel-Build} function recursively builds the finger lists for the left and right halves of the input in parallel, and then merges them. 
The base case is when there are fewer than $\alpha$ points, in which case we cannot even construct the first finger lists and can return.

During the merge phase,
we first construct the control tree using the control points of the left half, 
which can be done using a parallel semisort~\cite{gu2015top,dong2023highfull}. 
Then we process the control tree layer by layer in a BFS manner, starting from the roots. 
In \cref{fig:control-example}, we process $\{s_6,s_7\}$ in the first round, then $\{s_4,s_5\}$, and finally  $\{s_1,s_2,s_3\}$. 

For each point $i$ in the current layer, we resume its finger list construction from its control point $C[i]$ by the \fname{BuildFromRight}() function, which is similar to the sequential construction algorithm in \cref{alg:orig-build,alg:seq-efficient-build}. 
This allows us to directly advance to the right half in one hop, guided by $C[i]$, then we can continue the construction as normal. 
Specifically, we first set $F$ as the last complete finger list of $i$ so far, and will append new finger lists after it. 
We then compute the query radius $r$,
adjust the finger list to $\findlistradius{\cur}{r}$ based on the query radius $r$, 
check the condition to find the next focus point,
all using the same way as the sequential algorithm. 
Similarly, the random walk should terminate when we reach some $F^*=\findlistradius{\cur}{r}$,
where no closer point is found and $F^*$ is incomplete. 
In this case, the next focus point is also outside (to the right) the current subproblem. 
Based on the useful fact mentioned in \cref{sec:parallel-optimized}, the terminating point $\cur$ is the control point of $i$ in a future level of recursion. 
Therefore, we set $C[i]=\cur$ and terminate the process (line~\ref{line:par-control}).

In this algorithm, building the control tree takes linear work and $O(\log n)$ span, both \whp{}.
If we use binary search to find $\findlistradius{\cur}{r}$,
\cref{alg:par-build} directly gives us a parallel construction algorithm with $O(n \log{n}\log{\log{n}})$ work \whp{}.
The span bound is $O(\log^3 n)$ \whp{}: $O(\log n)$ levels of recursion, $O(\log n)$ levels in the control tree \whp{}, and the random walk in each invocation of \fname{BuildFromRight}.

\subsection{Maintaining the \emph{.\rghttxt{}} Pointers}\label{sec:map}
The final piece in our parallel algorithm is to efficiently maintain the \emph{.\rghttxt{}} pointers.
As such, we can achieve the $O(n\log{n})$ expected work bound as in our sequential algorithm.
To do this, we need to address two questions:
first, how to continue the random walk efficiently without all the \emph{.\rghttxt{}} being finished; and 
second, how to efficiently fill in the \emph{.\rghttxt{}} pointers once they become available. 
The answers, at a high level, are not complicated, and we will present them in this section. 
However, the algorithm with full details may be involved. Due to the space limit, we provide the full algorithm and in-depth explanation in \ifconference{the full version \cite{metricskiplistfullversion}}\iffullversion{Appendix.~\ref{sec:app:parallel-details}}.
\iffullversion{The pseudocode is given in \cref{alg:par-build-full}.}

For the first question, the case happens in the tail of the finger lists of the left subproblem.
When the points from the right subproblem are merged in, we need to fill in the points and set the \emph{.\rghttxt{}} pointers properly. 
To do this, we just start from the lowest level and use the \emph{\algntxt{}} function to move up and find the corresponding target. 
Due to the constant expansion rate, the cost of moving is asymptotically the same as the number of focus points from the right half. 
The solution for the second question, building the \emph{.\rghttxt{}} pointers, is slightly more complicated.
The \emph{.\rghttxt{}} pointers are computed from three sources:

\begin{itemize}
    \item Initialization: $\rght{\findlistrank{i}{1}}[j]$ points to $\findlistrank{j}{1}$,
    \item Push-down: $\rght{\findlistrank{i}{k}}[j]$ points to $\rght{\findlistrank{i}{k-1}}[j]$, 
    \item Shift focus: $\rght{\findlistrank{i}{k}}[j]$ points to $\rght{\somef{}}[j]$.
\end{itemize}
An \emph{\algntxt{}} is then called to calibrate the pointers to the correct radius.

The first two cases are simple as the base cases or can be processed by moving the pointers downward when possible.
The third case introduces new points into the finger lists and can change the radius, and in divide-and-conquer, it is possible that the target finger list is not ready but the random walk can continue and create this \emph{\rghttxt{}} pointer.
We note that all these \emph{\rghttxt{}} pointers created by ``shift focus'' align with the random walk, so when following the sequential execution, the dependences between all these pointers form a tree structure that resembles the control tree and also has $O(\log n)$ depth \whp{}.
We refer to this as the \emph{advance tree}. 
Hence, we can use the same approach to notify the updates of the \emph{.\rghttxt{}} pointers, just like notifying the points that their random walks are now available. 
We provide more details in \ifconference{the full version \cite{metricskiplistfullversion}}\iffullversion{Appendix.~\ref{sec:app:parallel-details}} on how the \rghttxt{}-tree is built and maintained.

\section{Applications}\label{sec:app}

Nearest neighbor search is widely used in real-world applications.
We can use our algorithm for the \msl{s} as the subroutine for these applications and achieve new parallel solutions.
For space limit, we discuss these applications in \ifconference{the full version of this paper \cite{metricskiplistfullversion}}\iffullversion{Appendix.~\ref{sec:app:app}},
including bichromatic closest pair (BCP)~\cite{agarwal1990euclidean}, density-based clustering~\cite{wang2020theoretically,Ester1996}, and $k$-NN graph construction. 

\section{Related Work}

\noindent\emp{NNS Algorithms.} \ Nearest-neighbor search has been studied
extensively in both Euclidean and general metric settings.  
Classical exact data structures include $k$-d trees~\cite{bentley1975multidimensional,men2025parallel} and other spatial trees~\cite{men2026parallel,gu2013efficient,finkel1974quad,kamel1992parallel,blelloch2022parallel}, 
navigating nets~\cite{krauthgamer2004navigating},
cover trees~\cite{beygelzimer2006cover,gu2022parallel}, \msl{}s~\cite{karger2002finding}, and others~\cite{clarkson1997nearest,nielsen2009bregman}.
Finding or even approximating nearest neighbor in the worst case requires linear time~\cite{abbasifard2014survey}; hence, known sublinear bounds are all under certain assumptions.
A complementary line of work focuses on approximate search, including locality-sensitive hashing, randomized space partitions, graph-based methods, and additional techniques such as quantization~\cite{manohar2024parlayann,peng2023efficient,li2019approximate,indyk1998approximate}.
Their primary goal is on efficient searches on high-dimensional vector data, while the goals differ from the exact NNS discussed in this paper.

The most relevant prior work is the cover tree, which also addresses theoretically efficient NNS in metric spaces. 
A comparison is presented in \cref{tab:cover-tree-msl}, and here we provide more details.
The cover tree was introduced by Beygelzimer et al.~\cite{beygelzimer2006cover} in 2006 with $O(n\log n)$ construction work and $O(\log n)$ query time under low expansion rate. 
Subsequent studies~\cite{curtin2015improving,elkin2021new} identified flaws in the original path compression technique, showing that the stated bounds required a bounded aspect ratio.
Gu et al.~\cite{gu2022parallel} in 2022 proposed the work-efficient parallel version of cover tree with $O(\log^3 n\log\log n)$ span, which remains subject to both low expansion rate and bounded aspect ratio assumptions.
Elkin and Kurlin~\cite{elkin2023new} in 2023 proposed a new sequential cover tree variant with the stated bounds without assuming bounded aspect ratio. 
To our knowledge, no parallel version of that variant exists, and the techniques by Gu et al.~\cite{gu2022parallel} do not readily generalize to it.
In contrast, our parallel algorithms for the \msl{} require only the low expansion rate assumption while achieving $O(n\log n)$ construction work, $O(\log^3 n)$ construction span, and $O(\log n)$ query cost.

Regarding practical implementations, sequential~\cite{beygelzimer2006cover} and parallel~\cite{izbicki2015faster} versions of cover trees exist; however, these implementations often relax the \emph{separation property}, sacrificing the original theoretical guarantees. 
While there are no known implementations of \msl{}s, applying our algorithmic insights to practical settings remains an interesting direction for future research.


\myparagraph{Applications.} Nearest-neighbor search is a fundamental primitive in computational geometry, supporting closest-pair problems~\cite{tao2010efficient,corral2000closest}, proximity graphs~\cite{iwasaki2018optimization}, and clustering~\cite{pourbahrami2022geometric}. In machine learning and data mining, it underlies nonparametric methods~\cite{cover1967nearest,altman1992introduction}, similarity search~\cite{indyk1998approximate}, anomaly detection~\cite{gu2019statistical,pang2015lesinn}, and various clustering techniques such as DBSCAN~\cite{Ester1996,wang2020theoretically} and spectral clustering~\cite{Maier2009,Franti2006,Lucinska2012,karypis1999chameleon,wang2021fast,yu2021parchain}. Furthermore, it is critical in vision, robotics, and database applications, including feature matching~\cite{garcia2010k}, motion planning~\cite{yershova2007improving}, and similarity joins~\cite{xiao2011efficient}. Since these tasks rely heavily on NNS or $k$-NNS, faster metric structures directly accelerate these downstream algorithms. Accordingly, this paper applies our results to obtain parallel bounds for bichromatic closest pair, density-based clustering, and $k$-NN graph construction under the constant expansion-rate assumption~\cite{karger2002finding}.

\section{Conclusion}

This paper presents a parallel algorithm for constructing metric skip-lists in $O(n \log n)$ expected work and $O(\log^3 n)$ span \whp{}. 
To achieve this, we first show a sequential solution with $O(n\log n)$ expected work.
The main challenge in parallelizing this algorithm lies in the inherent sequential construction and the complicated dependence structure that is hard to analyze. 
We overcome this by adding artificial synchronization barriers that convert the dependence structure to trees with bounded heights.
Putting all pieces together, we answered the decades-long open problem of how to achieve efficient nearest neighbor search on metric space with the only assumption of bounded expansion rate.


An interesting future work is to support updates in parallel or concurrent settings.
While the original paper \cite{karger2002finding,ruhl2003efficient} includes sequential updatable solutions,
extending them in the parallel setting would need careful design and analysis of the dependence structure, especially for the \emph{advance} pointers.
Another future direction is to support functional updates, which is needed in some applications such as the Euclidean minimum spanning tree (EMST) and the single-linkage clustering~\cite{gu2022parallel}.



\hide{
\section*{Acknowledgement}
This work is supported by NSF grants CCF-2103483, IIS-2227669, NSF CAREER Awards CCF-2238358 and CCF-2339310, the UCR Regents Faculty Development Award, and the Google Research Scholar Program.
}

\bibliographystyle{ACM-Reference-Format}
\balance
\bibliography{bib/strings,bib/main,local}

\iffullversion{
\appendix


\section{Applications}\label{sec:app:app}

Nearest neighbor search is widely used in machine learning algorithms.
We can use our algorithm for the \msl{s} as the subroutine for these algorithms and achieve highly parallel solutions.
Due to the space limit, we will overview a few classic applications, and we note that our \msl{s} can also be applied to other clustering algorithms such as the density peak clustering~\cite{huang2023faster}.

\subsection{Bichromatic Closest Pair (BCP)}

Given two sets $P_1$ and $P_2$, the bichromatic closest pair (BCP) is a pair of points $(p_1, p_2)$,
such that $p_1 \in P_1, p_2 \in P_2$, and $d(p_1, p_2) \le d(p'_1, p'_2) \mid \forall p'_1 \in P_1, \forall p'_2 \in P_2$.

WLOG, we assume $|P_1| = m \le n = |P_2|$.
We can construct a \msl{} for $P_1$, and query the nearest neighbor for every point in $P_2$ in parallel.
This gives $O(m \log n)$ expected work and $O(\log^3{n})$ span \whp{},
assuming constant expansion rate.

\subsection{Density-Based Clustering}

The density-based spatial clustering of applications with noise (DBSCAN) problem takes as input $n$ points $\mathcal{P} = \{p_0, \dots, p_{n-1}\}$,
a distance function $d$, and two parameters $\epsilon$ and minPts \cite{Ester1996}.
A point $p$ is a \textit{core point} if and only if $|B(p, \epsilon)| \ge \text{minPts}$.
We denote the set of core points as $C$.
DBSCAN computes and outputs subsets of $\mathcal{P}$, referred to as \textit{clusters}.
Each point in $C$ is in exactly one cluster, and two points $p, q \in C$ are in the same cluster if and only if there exists a list of points $\bar{p}_1 = p, \bar{p}_2, \dots, \bar{p}_{k-1}, \bar{p}_k = q$ in $C$ such that $d(\bar{p}_{i-1}, \bar{p}_i) \le \epsilon$.
For all non-core points $p \in \mathcal{P} \setminus C$, $p$ belongs to cluster $C_i$ if $p \in B(q, \epsilon)$ for any $q \in C \cap C_i$.
A non-core point belonging to at least one cluster is called a \textit{border point} and a non-core point belonging to no clusters is called a \textit{noise point}.

\citet{wang2020theoretically} recently showed how to parallelize DBSCAN based on bichromatic closest pairs (BCP).
Using the above results for BCP gives $O(n \log n)$ expected work and $O(\polylog(n))$ span \whp{} to compute DBSCAN.
Here the assumptions include: minPts and the expansion rate are constant, and a pairwise distance can be computed in constant time.

\subsection{$k$-NN Graph Construction}
$k$-NN graphs are widely used in machine learning, such as graph clustering~\cite{Maier2009,Franti2006,Lucinska2012,karypis1999chameleon}, manifold learning~\cite{Tenenbaum2000}, outlier detection~\cite{Hautamaki2004}, and proximity search~\cite{Paredes2005,Chavez2010,Sebastian2002}.
Given a point set $P$ in a metric space, a $k$-NN graph is a directed graph $G=(V,E)$, where $V=P$ and $(p,q)\in E$ if $q$ is one of $p$'s $k$-nearest neighbor in $V\setminus\{p\}$.
We first construct the \msl{} on $P$, then apply $(k+1)$-NN queries on all the points in $P$ in parallel,
and finally construct the $k$-NN graph according to the query results.
Using our parallel \msl{s}, we can get $O(kn\log k\log n)$ expected work and $O(\log^3 n)$ span \whp{}, assuming a constant expansion rate.

\section{More Details of the Work-Efficient Parallel Algorithm}\label{sec:app:parallel-details}

In \cref{sec:parallel} we give the high-level idea of the parallel construction.
Here we provide more details of the this algorithm.
\cref{alg:par-build-full} gives the pseudocode of the parallel construction of \msl{},
which is the extended version of \cref{alg:par-build}.
It has $O(n\log{n})$ expected work and $O(\log^3{n})$ span \whp{}.

\subsection{Step 1: Fill-in the Finger Lists in the Left Half}

Here, we define the control point of $i$ to be the last point that is added to $i$'s finger lists after \upshape{Parallel-Build}.
Note that this definition is slightly different from the one in \cref{sec:parallel},
as we may want to manually control the random walk sequence when building the finger lists for $i$.
However, the new definition still satisfies the analysis in \cref{sec:control-analysis},
so the height of the control tree is still $O(\log{n})$ \whp{}.

Like \cref{alg:par-build}, we process the control-tree layer-by-layer.
In each layer, we resume the random walk of $i$ in parallel.
For point $i$, we first jump to $\cur=C[i]$, and set $k$ as the last complete finger list of $\cur$,
then we keep decrementing $k$ until the radius meets the condition.
We can show that the cost of decrementing $k$ can be charged to the number of new points added to $\listoflist{i}$.

Starting from $\findlistrank{\cur}{k}$, we can continue the random walk process and fill in the finger lists for $i$ with points from the right half.
Note that we use $r=4\max(\radius{F},d(s_\cur,s_i))$ as the query radius, which is larger than the radius used in the sequential algorithm.
Doing this only requires increasing $\alpha$ by a constant factor.

At some point we encounter an incomplete finger list $\findlistrank{\cur}{k}$,
and we enter the second phase of Step 1.
This means that there are $O(\alpha)$ points in the ball of radius $B(i,3\radius{F})$ (see \cref{fig:4r-2r} for illustration).
We then process these $O(\alpha)$ points one by one to try to add them to $\listoflist{i}$.
Let $p_{1..t}$ be the points that are added to $\listoflist{i}$.
If we look at \cref{fig:4r-2r}, we can see that all these points are in the ball of radius $2\radius{F}$ from each other.
This means that we can manually control the random walk sequence:
from $i$ to $\cur$, and then to $p_1$, and then to $p_2$, and so on, finally to $p_t$.

\begin{figure}[h]
  \centering
  \includegraphics[width=0.5\columnwidth]{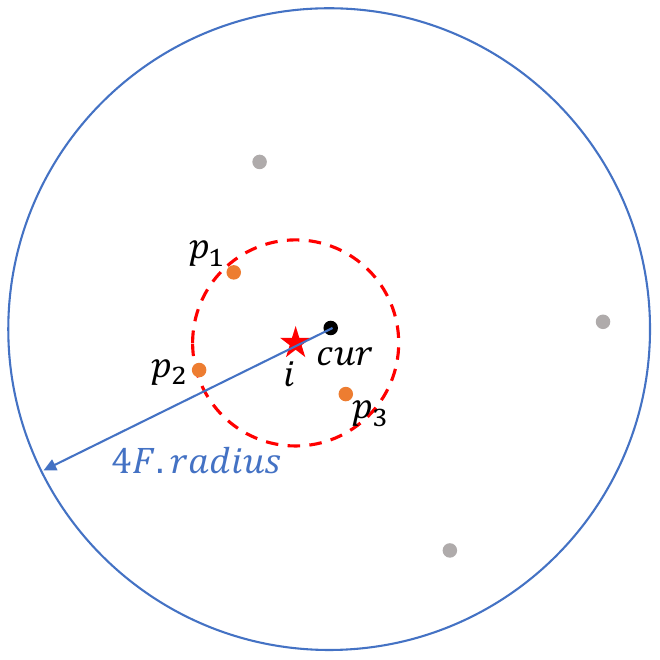}
  \caption{The reason to use $4\radius{F}$. This ensures that $p_{1..t}$ are all in the ball of radius $2\radius{F}$ from each other.
  }\label{fig:4r-2r}
\end{figure}

\subsection{Step 2: Maintaining the \rghttxt{} Pointers in the Left Half}

In the sequential construction algorithm \cref{alg:seq-efficient-build}, the \emph{\rghttxt{}} pointers are computed in three places: 
\begin{itemize}
    \item Initialization: $\rght{\findlistrank{i}{1}}[j]$ points to $\findlistrank{j}{1}$,
    \item Push-down: $\rght{\findlistrank{i}{k}}[j]$ points to $\rght{\findlistrank{i}{k-1}}[j]$,
    \item Shift focus: $\rght{\findlistrank{i}{k}}[j]$ points to $\rght{\somef{}}[j]$.
\end{itemize}
Then we try to align the $\rght{}$ pointers.
As mentioned in \cref{sec:map}, we can always maintain the $\rghttxt{}$-tree
using the dependencies between the $\rght{}$ pointers.

However, as we manually control the random walk sequence in Step 1,
for the last $O(\alpha)$ points $p_{1..t}$,
the $\rght{}$ pointers from $i$ to $p_{1..t}$ are not built yet, and they are not in the \rghttxt{}-tree.
Our key idea is that we manually maintain the $\rght{}$ pointers for the last $O(\alpha)$ points.
In the first step we add $p_{1..t}$ into a set $L[i]$.
In the second step, for each point $j\in L[i]$,
we start from the last complete finger list of $p_j$ and try to search $\rght{\findlistrank{i}{\cdot}}[j]$ and move up until we find a complete finger list.
The cost of this step can be also charged to the cost of building the finger lists for $i$ in Step 1, so it does not increase the overall work.
Note that the number of points in $L[i]$ is at most $\alpha$,
otherwise the first point in $L[i]$ will have a correct $\rght{}$ pointer from $i$ and will be removed from $L[i]$.

Now we have the \rghttxt{}-trees of the left half.
For point $j$, it can appear in the finger lists of $O(\log{n})$ points $i$,
so the size of the \rghttxt{}-tree rooted at $j$ is $O(\log{n})$.

In parallel for each $j$, we process its \rghttxt{}-tree sequentially, layer-by-layer, while enumerating the new finger lists of $j$.
If an \rghttxt{}-pointer becomes ready,
i.e. the radius of the finger list of $j$ is small enough,
we can mark this \rghttxt{}-pointer as ready, remove it from the \rghttxt{}-tree, and add its child \rghttxt{}-pointers to the frontier.
The overall process does not change the work of the sequential algorithm, while the span is $O(\log{n})$ because the \rghttxt{}-tree has $O(\log{n})$ size.

\begin{algorithm*}[h]
\DontPrintSemicolon
\small
\caption{Parallel-Build($s_{l \ldots r}$): Builds the finger lists and $\rght{}$ pointers for $s_{l \ldots r}$ in parallel}\label{alg:par-build-full}
\SetKw{Break}{break}

\SetKwInput{Maintains}{Maintains}
\Maintains{\\
$C[i]$: the control point for $s_i$, i.e. the last point added to $\listoflist{i}$ after Step 1\\
$L[i]$: the last $O(\alpha)$ points that are added to $\listoflist{i}$\\
$\parent{\cdot}$: the parent pointer for the \rghttxt{}-tree\\
}

$m \gets (l+r)/2$\\
\parDo{}{
Parallel-Build$(s_{l..m})$\\
Parallel-Build$(s_{m+1..r})$\\
}

\tcp*[h]{Step 1: fill-in the finger lists in the left half}\\
Construct the tree by $C[i]$ for each $i \in [l..m]$, in parallel \\

\ForEach{layer of the tree}{
    \parForEach{$i$ in this layer}{
        $\cur \gets C[i]$\\
        $k \gets$ the last complete finger list of $\cur$\\
        \upshape{Build-Phase-1}($i, \cur, k$)\\
        \upshape{Build-Phase-2}($i, \cur, k$)\\
        $C[i] \gets$ the last point added to $\listoflist{i}$\\
    }
}

\tcp*[h]{Step 2: maintain the \rghttxt{} pointers for the left half}\\
\parForEach{$i \in [l..m]$} {
    \ForEach{$j$ in $L[i]$}{
        Search $\rght{\findlistrank{i}{\cdot}}[j]$ from the last complete finger list of $p_j$ and move up\\
        \If{$\rght{\findlistrank{i}{\cdot}}[j]$ is ready}{
            Remove $j$ from $L[i]$\\
        }
    }
}
\parForEach{$j \in L$} {
    \tcp*[h]{Construct the $\rghttxt{}$ pointers pointing to $j$}\\
    \ForEach{new finger list $F$ added to $\listoflist{j}$}{
        Process the \rghttxt{}-tree and settle down as many \rghttxt{}-pointers as possible\\
    }
}

\medskip

\myfunc(){\upshape{Build-Phase-1}($i,\cur,k$)}{

\While{true}{
    $F \gets$ the last complete finger list of $s_i$\\
    $r \gets 4\max(\radius{F}, d(s_{i}, s_\cur))$\\
    
    \upshape{Align}$(\findlistrank{\cur}{k}, r)$\\

    \lIf{$|\findlistrank{\cur}{k}|<\alpha$}{
        \Break
    }

    \If{$\exists s_j\in \findlistrank{\cur}{k}$ such that $d(s_j, s_i) < \max(\radius{F}, d(s_\cur, s_i))$}{
        $\nxt \gets$ the first $j$ with this property in $\findlistrank{\cur}{k}$\\
        \If{$d(s_\nxt, s_i) < \radius{F}$}{
            Add a new finger list $F_\mathit{new}$ to $\listoflist{i}$\\
            $\parent{\node{\rght{F_\mathit{new}}[\nxt]}} \gets \node{\rght{\findlistrank{\cur}{k}}[\nxt]}$\tcp*[f]{maintaining the \rghttxt{}-tree}\\
        }
    }
    \Else{
        Let $\nxt$ be the last index in $\findlistrank{\cur}{k}$\\
    }
    $\cur \gets \nxt$\\
    $k \gets \rght{\findlistrank{\cur}{k}}[\nxt]$\\
}

}

\medskip

\myfunc(){\upshape{Build-Phase-2}($i,\cur,k$)}{

\While{true}{
    $F \gets$ the last complete finger list of $s_i$\\
    $r \gets 2\max(\radius{F}, d(s_\cur, s_i))$\\

    \upshape{Align}$(\findlistrank{\cur}{k}, r)$\\

    \If{$\exists s_j\in \findlistrank{\cur}{k}$ such that $d(s_j, s_i) < \radius{F}$}{
        $\nxt \gets$ the first $j$ with this property in $\findlistrank{\cur}{k}$\\
        Add a new finger list $F_\mathit{new}$ to $\listoflist{i}$\\
        Add $\nxt$ to $L[i]$\\
    }
    \lElse{
        \Break
    }
}

}

\end{algorithm*}

\section{Additional Proofs}\label{sec:app-proof}

\medskip
\begin{singleproof}[Proof of \cref{lem:4.1}]
  Fix a point $s_i$ and radius $r$. Let
\[
L=B(s_i,r),\quad U=B(s_i,2r),\quad m=|L|,\quad u=|U|.
\]
While building $\listoflist{i}$, only points in $U$ can affect finger lists whose radius is at most $2r$.
Now view points in $U$ in their (random) permutation order.

Starting from $\listoflist{i}(2r)$, a new list between radii $2r$ and $r$ can only be created before we have seen $\alpha$ points from $L$:
once $\alpha$ points from $L$ have appeared, the maintained set of $\alpha$ closest seen points is contained in $L$, so the current radius is at most $r$, i.e., we have reached $\listoflist{i}(r)$.
Hence, the number of such intermediate lists is at most the number of draws from $U$ until the $\alpha$-th point of $L$ appears.

This is a negative-hypergeometric stopping time. Its expectation is
\[
\mathbb{E}[T]=\frac{\alpha(u+1)}{m+1}=O\left(\alpha\frac{u}{m}\right).
\]
By the expansion condition (for the radii where it applies), $u=|B(s_i,2r)|\le c\,|B(s_i,r)|=cm$.
Therefore,
\[
\mathbb{E}[T]=O(\alpha c)=O(1),
\]
since $\alpha,c$ are constants.
So the expected number of lists between $\listoflist{i}(r)$ and $\listoflist{i}(2r)$ is $O(1)$.
\end{singleproof}

\bigskip

\begin{singleproof}[Proof of \cref{lem:4.2}]
  The random walk (loop on line~\ref{line:seq-query-main}) will finish in $O(\log n)$ iterations \whp{}.
  Hence, line~\ref{line:up} will be invoked for $O(\log n)$ iterations \whp{}.
  Note that line~\ref{line:up} will execute only when the search radius grows in this iteration, when the last points in $F^*$ is selected (on line~\ref{line:seq-last}), and the radius can grow at most twice.
  Let $r'$ be the search radius in the previous iteration, and $s_{\last}$ be the last point in the random work ($\bar{r'}$ as its finger list radius after aligning on line~\ref{line:seq-align}).
  $\findlistrank{\cur}{k}$ is the first list of $s_\cur$ of the radius no more than $\bar{r'}$.
  We need line~\ref{line:up} to increase the radius from $\bar{r'}$ to $r$.
  Note that since $s_\cur$ is in $\listoflist{s_{\last}}$, $\bar{r'}\ge d(s_\cur, s_{\last})$ holds, so as $r\le 2\bar{r'}$.
  According to \cref{lem:4.1}, the cost of this step is $O(1)$.
\end{singleproof}

\begin{singleproof}[Proof of \cref{lem:4.3}]
Each execution of \emph{align} consists of two while loops: upward moves (\(k\gets k-1\), line~\ref{line:up}) and downward moves (\(k\gets k+1\), line~\ref{line:align-end}).
So it suffices to bound the total numbers of these two moves over one call of \cref{alg:seq-find}.

The total number of upward moves is \(O(\log n)\) in expectation by Lemma 4.2.
Now consider downward moves.
Define a potential at the beginning of each outer iteration:
\[
\Phi=\log |B(s_{\cur},r)|,\quad r=2d(s_{\cur},q).
\]
By the finger-list property used in the paragraph below Theorem 4.4, after \(O(1)\) new elements in the list, the ball size is expected to halve; equivalently, every \(O(1)\) downward \emph{align} moves decrease \(\Phi\) by \(\Theta(1)\) in expectation.
Hence, the number of downward moves is proportional (in expectation) to the total decrease of \(\Phi\), up to additive increases of \(\Phi\).

There are only two sources of increase:
first, transitioning \(\somef{}\to \rght{\somef{}}[j]\), which changes the radius only by a constant factor, so \(\Phi\) changes by \(O(1)\);
second, the explicit upward \emph{align} moves, already bounded by Lemma 4.2.
Therefore \(\log|B(\cur,r)|\) behaves as the same negative-drift random walk used in the query analysis (\cref{sec:intuition,sec:kr}), and over the whole query its total expected increase is \(O(\log n)\).
Since initially \(\Phi\le \log n\), the total expected downward moves are \(O(\log n)\).

Combining both directions, total \emph{align} work is
\[
O(\#\text{up}+\#\text{down})=O(\log n)+O(\log n)=O(\log n)
\]
in expectation.
\end{singleproof}

}

\balance

\end{document}